\documentclass[11pt]{article}

\usepackage{amsmath}

\usepackage{times}
\usepackage{t1enc}
\usepackage{dsfont}
\usepackage{mathtools}

\usepackage{amsfonts,amsmath,amsthm}
\usepackage{amssymb}
\usepackage{paralist}
\usepackage[pdfdisplaydoctitle,pagecolor=orange!40!black,menucolor=orange!40!black,filecolor=magenta!40!black,urlcolor=blue!40!black,linkcolor=red!40!black,citecolor=green!40!black,colorlinks]{hyperref}
\usepackage{multirow}
\usepackage{xspace}
\usepackage{graphicx}
\usepackage{pgf,tikz}
\usetikzlibrary{positioning}
\usepackage[textsize=scriptsize]{todonotes}

\usepackage[numbers]{natbib}
\usepackage{hyperref}
\usepackage[noend,ruled]{algorithm2e}
\usepackage{colortbl}
\usepackage{lscape}
\usepackage{enumitem}
\usepackage{comment}
\usepackage{float}
\usepackage{url}
\usepackage{caption}
\usepackage{subcaption}
\usepackage{pgfplots}
\pgfplotsset{}

\allowdisplaybreaks

\newtheorem{theorem}{Theorem}

\newtheorem{construction}{Construction}

\newtheorem{definition}{Definition}
\newtheorem{proposition}{Proposition}

\usepackage[]{cleveref}
\usepackage{url}

\newcommand{\pref}{\succ}
\newcommand{\fixlist}{\addtolength{\itemsep}{-2pt}}
\newcommand{\score}{{{\mathrm{score}}}}

\crefname{subsection}{Section}{Sections}
\crefname{section}{Section}{Sections}

\crefname{table}{Table}{Tables}
\crefname{figure}{Figure}{Figures}
\crefname{algorithm}{Algorithm}{Algorithms}
\crefname{theorem}{Theorem}{Theorems}
\crefname{definition}{Definition}{Definitions}
\crefname{corollary}{Corollary}{Corollaries}
\crefname{proposition}{Proposition}{Propositions}
\crefname{obs}{Observation}{Observations}
\crefname{lemma}{Lemma}{Lemmas}
\crefname{example}{Example}{Examples}
\crefname{reduction}{Reduction}{Reductions}
\crefname{algocf}{Algorithm}{Algorithms}

\newcommand{\calR}{{{\mathcal{R}}}}
\newcommand{\calS}{{{\mathcal{S}}}}

\newcommand{\np}{{\mathsf{NP}}}
\newcommand{\fpt}{{\mathsf{FPT}}}

\newcommand{\paranp}{{\mathsf{Para}\textrm{-}\mathsf{NP}}}

\newcommand{\xp}{{\mathsf{XP}}}
\newcommand{\wone}{{\mathsf{W[1]}}}

\newcommand{\p}{{\mathsf{P}}}

\newcommand{\MIPlong} {\textsc{Mixed Integer Linear Programming}\xspace}
\newcommand{\MIP} {\textsc{MIP}\xspace}
\newcommand{\eMIPlong} {\textsc{Mixed Integer Programming with Simple Piecewise Linear Transformations}\xspace}
\newcommand{\eMIP} {\textsc{MIP/SPLiT}\xspace}

\newcommand{\msmc}     {\textsc{Multiset Multicover}\xspace}
\newcommand{\wmsmc}    {\textsc{Weighted Multiset Multicover}\xspace}

\newcommand{\piece}{{{\ensuremath{\rho}}}}
\newcommand{\der}{{{\mathrm{der}}}}
\newcommand{\pieces}{{{\mathrm{pieces}}}}
\newcommand{\miss}{{{\mathrm{miss}}}}
\newcommand{\reals}{{{\mathbb{R}}}}

\newcommand{\pmax}{\ensuremath{p_{\max}}}

\newcommand{\probDef}[3]{
  \begin{quote}
   #1\\
  \textbf{Input:} #2\\
  \textbf{Question:} #3
  \end{quote}
}

\title{Mixed Integer Programming with Convex/Concave Constraints:
  Fixed-Parameter Tractability and Applications to Multicovering and
  Voting\footnote{A preliminary version of this paper appeared under
    the title ``Elections with Few Candidates: Prices, Weights, and
    Covering Problems'' in \emph{Proceedings of the 4th International
      Conference on Algorithmic Decision Theory, ADT
      2015}~\cite{BFNST15}. This journal version focuses on the
    theoretical part which is substantially revised and improved.}}

\author{
  \makebox[0.25\linewidth]{Robert Bredereck} \\ 
  TU Berlin \\
  Berlin, Germany
\and
  \makebox[0.25\linewidth]{Piotr Faliszewski} \\ 
  AGH University \\ 
  Krakow, Poland
\and 
  \makebox[0.25\linewidth]{Rolf Niedermeier} \\
  TU Berlin \\
  Berlin, Germany
\and 
  \makebox[0.25\linewidth]{Piotr Skowron} \\
  TU Berlin \\
  Berlin, Germany
\and 
  \makebox[0.25\linewidth]{Nimrod Talmon} \\
  Weizmann Institute of Science \\
  Rehovot, Israel
}

\date{}

\begin{document}

\maketitle

\setcounter{footnote}{0}

\begin{abstract}
  A classic result of Lenstra [Math.~Oper.~Res.~1983] says that an
  integer linear program can be solved in fixed-parameter tractable
  ($\fpt$) time for the parameterization by the number of variables.
  We extend this result by incorporating piecewise linear convex or
  concave functions to our (mixed) integer programs.  This general
  technique allows us to analyze the parameterized complexity of a
  number of classic $\np$-hard computational problems.  In particular, we prove
  that \textsc{Weighted Set Multicover} is in $\fpt$ when
  parameterized by the number of elements to cover, and that there
  exists an $\fpt$-time approximation scheme for \textsc{Multiset
    Multicover} for the same parameter.  Further, we use our general
  technique to prove that a number of problems from computational
  social choice (e.g., problems related to bribery and control in
  elections) are in $\fpt$ when parameterized by the number of
  candidates.  For bribery, this resolves a nearly 10-year old family
  of open problems, and for weighted electoral control of Approval
  voting, this improves some previously known $\xp$-memberships to
  $\fpt$-memberships.
\end{abstract}

\section{Introduction}

The idea of parameterized complexity theory is to measure the
difficulty of computational problems with respect to both the length
of their encodings (as in standard complexity theory) and
additional parameters (e.g., pertaining to the structure of the
input). For example, in the \textsc{Set Multicover} problem we are
given a set of elements $U = \{x_1, \ldots, x_m\}$, a multiset $\calS
= \{S_1, \ldots, S_n\}$ of sets over~$U$, integer covering
requirements $r_1, \ldots, r_m$ for the elements of $U$, and a budget
$B$.  The question is whether it is possible to pick a collection of
at most $B$ sets from $\calS$ so that each element $x_i \in U$ belongs
to at least $r_i$ of them.  This problem is a simple extension of the
\textsc{Set Cover} problem and, thus, is $\np$-hard.
There exists, however, an efficient algorithm for this problem
when the cardinality of $U$ is small. In
fact, \textsc{Set Multicover} is fixed-parameter tractable (is in $\fpt$)
with respect to the number of elements to cover:
We have an algorithm that solves it in time $f(|U|)\cdot |I|^{O(1)}$,
where $|I|$ is the length of the encoding of the given instance and
$f$ is a computable function (which depends only on the parameter;
i.e., $|U|$ in our case).



This $\fpt$ algorithm for \textsc{Set Multicover} proceeds by
expressing the problem as an integer linear program (ILP) and solving
it using the classic algorithm of Lenstra~\cite{Len83}.  For each
subset~$A$ of~$U$, let $\calS(A)$ be the subfamily of $\calS$ that
contains sets equal to $A$; use a variable $x_A$, intended to hold
the number of sets from $\calS(A)$ that we include in the solution;
and define a constant $b_A = |\calS(A)|$. Then, we introduce the
following constraints:
\begin{align}
\label{c1}  x_A \leq b_A  &\quad \text{for each $A \subseteq U$},\\
\label{c2}  \sum_{A \subseteq U} x_A \leq B, & \\
\label{c3}  \sum_{\mathclap{A \subseteq U\colon x_i \in A}} x_A \geq r_i & \quad \text{for each $x_i \in U$.}
\end{align}
Constraints of the form~\eqref{c1} ensure that the solution is
possible (i.e., we never use more sets of a given type than there are
in the input), constraint~\eqref{c2} ensures that we use at most $B$
sets, and constraints of the form~\eqref{c3} ensure that each element
from $U$ is covered a required number of times.  We solve this ILP
using the algorithm of Lenstra~\cite{Len83}, which decides feasibility
of ILPs in $\fpt$ time with respect to the number of integer
variables.

Unfortunately, the above approach seems to fail for the case of
\textsc{Weighted Set Multicover}, a variant of \textsc{Set Multicover}
where each set from $\calS$ also has a weight and we can only choose
sets of total weight at most $B$ (\textsc{Set Multicover} is the special
case of \textsc{Weighted Set Multicover} where all the weights are
equal to one). The reason for this apparent failure is that, for each
$A\subseteq U$, the ILP for \textsc{Set Multicover} treats all sets from
$\calS(A)$ as indistinguishable, but in \textsc{Weighted Set
  Multicover} they have weights, which give them ``identities.''
Specifically, for the case of \textsc{Weighted Set Multicover} we
would have to replace constraint~\eqref{c2} with one of the form:
\begin{align*}
 \sum_{A \subseteq U} \mathrm{weight}_A(x_A)  \leq B,
\end{align*}
where $\mathrm{weight}_A(x)$ is a (convex) function that gives the sum
of the lowest $x$ weights of the sets from~$\calS(A)$. The main
contribution of this paper is in showing a technique that allows us to
include constraints of this type, where some variables are replaced by
their ``piecewise linear convex or concave
functions'', and still solve the resulting programs in $\fpt$ time
using Lenstra's algorithm.

In \autoref{sec:ilp} we compare our technique to similar tools known in the literature. 
We argue that one of the greatest benefits of our technique is that it is
very simple to use. Indeed, to obtain an $\fpt$ algorithm for
\textsc{Weighted Set Multicover} it suffices to make only very small
changes to the algorithm for \textsc{Set Multicover} and invoke our
machinery. Second, our technique proceeds by constructing a mixed integer program
from an integer linear program with aforementioned convex constraints---this 
makes it easy to feed it into existing commercial solvers and to obtain an efficient practical algorithm with a relatively low effort.
Finally, we consider that the main contribution of this work is that we illustrate the simplicity and usefulness of the
discussed technique by applying it to a number of set-covering problems and by
resolving the computational complexity status of a number of
election-related problems parameterized by the number of candidates.
These problems include, for example, various bribery
problems~\cite{EFS09,FHH09} and priced control problems~\cite{MF16}
that were known to be in~$\xp$ for nearly ten years, but were neither
known to be fixed-parameter tractable, nor to be
$\wone$-hard. Indeed, in all these problems the voters had prices,
which gave them ``identities'' in the same way as weights gave
``identities'' to sets in the \textsc{Weighted Set Multicover}
problem.
Nonetheless, our technique is not limited to the case of election
problems with prices. For example, we show that it also applies to
weighted voter control for Approval voting, improving results of
Faliszewski et al.~\cite{FHH13}, and Faliszewski et
al.~\cite{fal-sko-sli-tal:c:top-k-counting} apply our technique to
problems pertaining to finding winners according to several
multiwinner election rules (Peters~\cite{Pet17a} and Caragiannis~et~al.~\cite{CKMPSW16} use a very similar
technical trick).
%


We also demonstrate the usefulness of our technique in more
technically demanding scenarios.  In particular, we consider the
\textsc{Multiset Multicover} problem, which generalizes \textsc{Set
  Multicover} to the case of multisets, and we show that our general
technique can be used as a component in the construction of an
$\fpt$-time approximation scheme for the problem (parameterized 
by the number of elements, there is no hope for
an FPT exact algorithm as the problem is $\np$-hard even for two
elements).  In this case, our analysis combines combinatorial
arguments with the technique of handling ILPs with piecewise-linear
convex/concave transformations.

\subsection{Related Work}

Our work is related to three main research lines, regarding algorithms
for integer linear programming, regarding various types of
covering problems, and regarding algorithms for and complexity of
voting-related problems. Below we review the main points of
intersection with this literature.

\subsubsection{Integer Linear Programming}\label{sec:ilp}
Solving integer linear programs is one of the major approaches for
tackling $\np$-hard problems~\cite{Sch98}. Indeed, modern ILP solvers
such as CPLEX or Gurobi can effectively deal with ILPs with
thousands of variables and constraints; this makes them practical
tools to solve even fairly large instances of computationally
difficult problems.
Yet, from our point of view, the most important aspect of integer
linear programming is that it provides a powerful tool for designing
$\fpt$ algorithms. The first major breakthrough in this direction was
achieved by Lenstra~\cite{Len83}, who showed that \MIPlong is
fixed-parameter tractable with respect to the number~$n$ of the
variables. Then \citet{FT87ilp} and \citet{Kan87} improved the
corresponding running time bounds.

The algorithm of Lenstra is very useful, but for some problems its
parameterization by the number of integer variables seems insufficient
to obtain $\fpt$ algorithms (the \textsc{Swap Bribery}
problem~\cite{EFS09} parameterized by the number of candidates is one
particular example where this seems to be the case). Fortunately, one
may use other approaches, such as the recently popularized $n$-fold
integer programming technique ($n$-fold IP). The main idea of $n$-fold
integer programming is as follows: Our goal is to find a feasible
solution of an inequality of the form $E\cdot x \leq b$, where $E$ is
a matrix, $b$ is a vector of constants, and $x$ is a vector of integer
variables (all as in the classic ILP problem), but where $E$ is
restricted to be of a very special form ($A$ and $D$ are matrices):
\begin{equation*}
  E = \left(
      \begin{array}{ccccc}
        D & D & D & \cdots & D \\
        A & 0 & 0 & \cdots & 0 \\
        0 & A & 0 & \cdots & 0 \\
        0 & 0 & A & \cdots & 0 \\
        \vdots & \vdots & \vdots & \ddots & 0 \\
        0 & 0 & 0 & \cdots & A \\
      \end{array}
      \right).
\end{equation*}
There is an algorithm that solves ILPs of this form in $\fpt$ time
parameterized by the dimensions of the matrices $A$ and
$D$~\cite{del-hem-onn-wei:j:n-fold-ip,del-hem-kop:b:disciret-optimization-n-fold,hem-onn-rom:j:n-fold-ip},
without restricting the dimensions of matrix $E$ (i.e., the
constraints have a restricted form, but we can use as many integer
variables as we like). In spite of this restrictive structure, $n$-fold integer
programming found a number of applications, e.g., in
scheduling~\cite{kno-kou:t:n-fold-scheduling} and voting~\cite{KKM17}.
See also the works of Dvor{\'{a}}k et al.~\cite{DEGKO17} and Knop et al.~\cite{KKM17b}
for very recent generalizations and extensions of the this technique.

It is quite non-obvious how the technique of $n$-fold integer
programming compares to ours.
On the one hand, using $n$-fold IP it is
possible to find $\fpt$ algorithms for problems for which our approach
seems not to be applicable (such as \textsc{Swap Bribery} parameterized
by the number of candidates~\cite{KKM17}). On the other hand, it is
not clear if $n$-fold IP can be used in all cases where our method
works (but we also cannot provide an obvious example where it fails
and our approach works; as a side note, it seems as if proving its failure is an interesting hard task).
While we stress that our approach is easier to use than that of
$n$-fold integer programming, we also point out that Kouteck\'y et
al.~\cite{KKM17} made significant progress in making $n$-fold IP more
approachable. In particular, they described how to implement a number
of primitives, and showed how using them makes the task of building $n$-fold
programs much easier compared to directly designing the corresponding $A$ and $D$ matrices.

Further, we note that more general techniques than the one presented
in this paper are known in the literature\footnote{We thank Martin
  Kouteck\'{y} for pointing out to us the most relevant works in this
  area.}. For instance, Dadush~et~al.~\cite{DPV11} proved that an ILP
can be solved in $\fpt$ time with respect to the number of integer
variables, even if the constraints describe an arbitrary convex
polyhedron, and are given through a separation oracle.  One particular
aspect in which the result of Dadush~et~al.~\cite{DPV11} generalizes
our technique is that it allows for using convex \emph{multi-variate}
functions in the constraints formulations.  (A similar argument has
been also given by Hildebrand and K{\"{o}}ppe~\cite{HK13}---for the
case when constraints are expressed as convex polynomials, and by
Khachiyan and Porkolab~\cite{KhaPor00}---for the case when the
constraints are expressed as convex semialgebraic sets.\footnote{Some
  of these results are not commonly known among the community centered
  around parameterized complexity theory.  For instance, Khachiyan and
  Porkolab claim to show a polynomial time algorithm for solving ILPs
  provided that the number of integer variables is constant; thus,
  this result appears to show an $\xp$ membership only.  However, a
  more careful look at the paper allows to realize that Khachiyan and
  Porkolab in fact give an $\fpt$ algorithm.  Again, we thank Martin
  Kouteck\'{y} for fruitful discussions on clarifying these issues.})
Gaven\v{c}iak~et~al.~\cite{GKK17} give a comprehensive review of the
advances in solving convex integer programs from the last two decades.

Yet, we believe that our approach has two advantages that make it
preferable whenever it can be used. First, the technique explained by
our proof is simpler, uses only elementary methods, and can be
easily understood without the necessity of learning quite advanced
tools.  Second, our technique reduces ILP programs with convex
constraints to standard MILP programs, which allows to easily use
cutting edge off-the-shelf MILP solvers with the more general types of
constraints.  Thus, due to our results one may obtain both a
theoretical $\fpt$ guarantee and a practical algorithm.


\subsubsection{Covering}
The class of covering problems is of fundamental importance in
theoretical computer science because, on the one hand, covering
problems are abstractions of many real-life issues, and, on the other
hand, they exhibit very interesting algorithmic behavior.  For
example, the \textsc{Set Cover} problem, arguably the best known
representative of the class, was among the first problems shown to be
$\np$-complete (in Karp's seminal paper~\cite{Kar72}), and later it was
thoroughly studied from the points of view of
(in)approximability~\cite{WS11} and
parameterized complexity~\cite{DF13}. For example, it is known to be
W[2]-hard for the parameterization by the solution size (indeed, it is
one of the classic W[2]-hard problems), but it is fixed-parameter
tractable with respect to the number of elements.  Covering problems
have various applications in domains such as software engineering
(e.g., covering scenarios by few test cases), antivirus development
(looking for a set of suspicious byte strings which covers all known
viruses), databases (finding a set of labels which covers all data
items), to name just a few.

There is a vast literature on the \textsc{Set Cover} problem and,
thus, we only briefly point out selected results.  It is known that a
simple greedy algorithm gives a $\log(m)$ approximation guarantee,
where $m$ is the size of the set to be covered (this algorithm was
given, e.g., by Johnson~\cite{Joh74}, but see the textbook of
Vazirani~\cite{Vaz01} for further references). It is also known that
unless $\p = \np$, no polynomial-time algorithm can approximate the
problem with a better ratio~\cite{Fei98,din-ste:c:set-cover}.  The
variant of the problem where each element appears in at most $f$ sets
can be approximated with the ratio~$f$~\cite{Vaz01}.  
Parameterized approximation algorithms for the problem were considered
by Bonnet et~al.~\cite{journals/corr/BonnetPS13}, Skowron and
Faliszewski~\cite{sko-fal:c:maxcover} and Skowron~\cite{Skow16a}.

\textsc{Set Multicover} and \textsc{Multiset Multicover} also received
some attention in the literature. For example, Berman et
al.~\cite{ber-das-son:j:set-multicover} considered approximability of
the former problem and Rajagopalan and
Vazirani~\cite{raj-vaz:j:multicover} studied the same issue for the
latter (and, in general, for various covering problems).  Exact
algorithms for these problems were studied by Hua
et~al.~\cite{conf/isaac/HuaYLW09}. Approximation algorithms for
covering integer programs were considered by
Kolliopoulos~\cite{Kolliopoulos03Covering} and Kolliopoulos and
Young~\cite{Kolliopoulos05Covering}.

\subsubsection{Voting}

Computational social choice (COMSOC) is an interdisciplinary area that
spans computer science, economics, and operations research, and whose
goal is to (computationally) analyze group decision-making
processes~\cite{Rot16,BCELP16}. From our point of view, the most
relevant part of COMSOC regards the complexity of various
election-related problems. We mostly focus on the problems of
manipulating election outcomes through bribery and
control~\cite{fal-rot:b:control-bribery}.


In the simplest variant of the bribery problem~\cite{FHH09}, we are
given a description of an election (i.e., a set of candidates and a
collection of voters, with their votes represented in some way) and an
integer $k$. We ask if it is possible to ensure that a given candidate
becomes a winner by bribing at most some $k$ voters (i.e., by changing
the votes of at most $k$ voters). However, there are many other
variants where, for example, each voter may have a different price for
being bribed~\cite{FHH09}, the prices may depend on the extent to
which we change given votes~\cite{EFS09}, or where the votes are
represented using some involved language~\cite{NMFK12Z}. Initially,
bribery problems were supposed to model undesirable behavior, but
later researchers realized that they also capture perfectly legal
actions, such as campaigning~\cite{BFLR12,BCFNN16,DS12,SFE17}, fraud
detection~\cite{Xia12}, or analysis of candidate
performance~\cite{fal-sko-tal:c:bribery-measure-success}.  Control
problems are similar in spirit to the bribery ones, but instead of
modifying votes, we are allowed to add/delete either candidates or
voters~\cite{BTT92,HHR07}.

We focus on the case where we have a few candidates but (possibly)
many voters.  
This is a very natural setting and it models many real-life scenarios
such as political elections 
or elections among company stockholders.  The complexity of
manipulating elections with few candidates is, by now, very well
understood. On the one hand, if the elections are weighted (as is the
case for the elections held by company stockholders), then our
problems are typically $\np$-hard even if the number of candidates is
a small fixed constant~\cite{CSL07,FHH09,FHH13}; these results
typically follow by reductions from the well-known $\np$-hard
\textsc{Partition} problem.  One particular example where we did not
have $\np$-hardness for the setting with the fixed number of candidates
was control by adding/deleting
voters under the Approval and $k$-Approval voting rules. For this problem
Faliszewski at al.~\cite{FHH13} have shown $\xp$ membership, but could
neither prove W[1]-hardness nor give an $\fpt$-algorithm; in this
paper we resolve this problem by proving fixed-parameter tractability.

If we consider parameterization by the number of candidates but the
elections are unweighted (as is the case for political elections) and
no prices are involved, then we typically get $\fpt$ results.  These
results are often obtained by expressing the respective problems as
integer linear programs (ILPs) and then applying Lenstra's
algorithm~\cite{Len83}.
In essence, these results are of the same nature as the $\fpt$
algorithm for the \textsc{Set Multicover} problem given at the
beginning of the introduction.
The main missing piece in our understanding of the complexity of
manipulating elections with few candidates regards those
unweighted-election problems where each voter has some sort of price
(for example, as in the bribery problems). In this paper we almost
completely fill this gap by showing a general approach for proving
$\fpt$ membership for a class of bribery-like problems parameterized
by the number of candidates, for unweighted elections.

\subsection{Organization}

The paper is organized as follows. First, in
\autoref{sec:general_technique}, we describe our technique of
handling ILPs with piecewise-linear convex/concave transformations.
Then, in \autoref{sec:showcase}, we show examples of applying it
to covering problems and to problems regarding Approval voting (which
are very close in spirit to covering problems). In
\autoref{sec:generalizations} we show how further election-related
problems, beyond Approval voting, can be solved using our approach, and
we conclude in \autoref{sec:conclusions}.

\section{Mixed Integer Linear Programming with Piecewise Linear Convex/Concave Functions}\label{sec:general_technique}


We now describe our technique of solving
integer linear programs with piecewise linear convex/concave
transformations. The technique is based on using rational-valued
variables to simulate the behavior of convex/concave transformations
and it uses as a subroutine an algorithm for solving the \MIPlong problem, as defined below.

\probDef{\MIPlong (\MIP)} 
  {An $m \times (n+r)$ matrix~$A$ with integer elements and a length-$m$
  integer vector~$b$ of rational numbers.}
{Is there a length~$(n+r)$ vector~$x=(x_1,\ldots ,x_{n+r})$ such that
  $A \cdot x \le b$, where values $x_1, \ldots, x_n$ are required to
  be integers, but values $x_{n+1}, \ldots, x_{n+r}$ can be rational?}

\noindent
We use standard terminology and syntax from linear programming.  In
particular, we interpret the entries of vector~$x$ as \emph{variables}
and the rows of matrix~$A$ as \emph{constraints}. Formally, we require
matrix~$A$ and vector~$b$ to have integer entries, but it would be
straightforward to allow them to have rational values (indeed, we
implicitly assume that whenever we obtain matrix $A$ or vector $b$
with rational values, they are transformed to the integer form by
appropriate multiplication).

Formally, the \MIP problem captures the issue of testing whether
\emph{some} feasible solution exists for a given mixed integer
program.  If we want a solution that maximizes a certain objective
function of the form $c_1x_1 + c_2x_2 + \cdots + c_{n+r}x_{n+r}$ (where $c_1,
\ldots, c_{n+r}$ are integers), then we can use the standard trick of
including the constraint:
\[
   c_1x_1 + c_2x_2 + \cdots + c_{n+r}x_{n+r} \geq T
\]
in the program and performing a binary search for the largest value of
$T$ for which a feasible solution exists. If $M$ is the value of the
objective function for the optimal solution, then this requires solving
$\log M$ programs instead of one. Typically this is a perfectly
acceptable price to pay.

The following result, due to \citet{FT87ilp} and \citet{Kan87}, gives
the $\fpt$ algorithm for the \MIP problem parameterized by the number
of integer variables (the result is built on top of Lenstra's original
algorithm~\cite{Len83}).

\begin{theorem}[Lenstra~\cite{Len83}, Frank and Tardos~\cite{FT87ilp},
  and Kannan~\cite{Kan87}]\label{thm:mip}
  There is an algorithm that solves \MIPlong in $O(n^{2.5n
    +o(n)}\cdot |I|)$ time, where $|I|$~is the number
  of bits encoding the input and $n$ is the number of integer variables.
\end{theorem}


Below we describe how the \MIP problem and the above theorem can be
extended to the case where each variable is replaced by some piecewise
linear convex/concave transformation, while still maintaining the
$\fpt$ running time.


\paragraph{Piecewise Linear Convex/Concave Transformations.}

We consider two simple types of piecewise linear transformations,
\emph{piecewise linear convex functions} and \emph{piecewise linear
  concave functions}.  A piecewise linear convex function is a
continuous convex function, defined on a finite sequence of intervals
(that together give the set of all real numbers) so that for each of
the intervals, the function restricted to this interval is linear.
Piecewise linear concave functions are defined analogously, except
that they are concave instead of convex. We present examples of functions
of this type in \autoref{fig:piecewise_linear}.

For a piecewise linear convex/concave function $f: \reals \to \reals$,
we denote the decomposition of its domain into a minimal number of
disjoint intervals on which it is linear as follows:
\begin{align*}
\reals =\; (-\infty, \piece(f, 1)] \cup 
           (\piece(f, 1), \piece(f, 2)] \cup
           (\piece(f, 2), \piece(f, 3)] \cup \ldots \cup
           (\piece(f, \ell), \infty) \text{.}
\end{align*}
We say that such a function $f$ consists of $\ell+1$ pieces and we
define $\pieces(f)$ to be $\ell+1$.  We refer to the function~$f$
restricted to interval~$(-\infty, \piece(f, 1)]$ as the zeroth piece
of~$f$, to the function~$f$ restricted to interval $[\piece(f, 1),
\piece(f, 2)]$ as to the first piece of $f$, to the function~$f$
restricted to interval $[\piece(f, 2), \piece(f, 3)]$ as to the second
piece of $f$, and so on.  By $\der(f, i)$ we denote the (constant)
derivative of the $i$-th piece of $f$.

For technical reasons we also require that for each piecewise linear
convex/concave function we have that
(a) $f(0)$ is integer,
(b) for each piece $i$, $0 \leq i \leq \pieces(f)$, the derivative $\der(f,i)$ is integer, and
(c) for each $i$, $1 \leq i \leq \pieces(f)-1$, $\piece(f,i)$ is an integer.
These requirements are technical only and, for example, we could
replace all the occurrences of the word ``integer'' with ``rational
value'' and all our results would still hold (this is analogous to the
fact that we could define the \MIP problem to work with rational
values in the matrix $A$ and vector $b$).

\begin{figure}[t!]
  \begin{center}
    \includegraphics[scale=0.48]{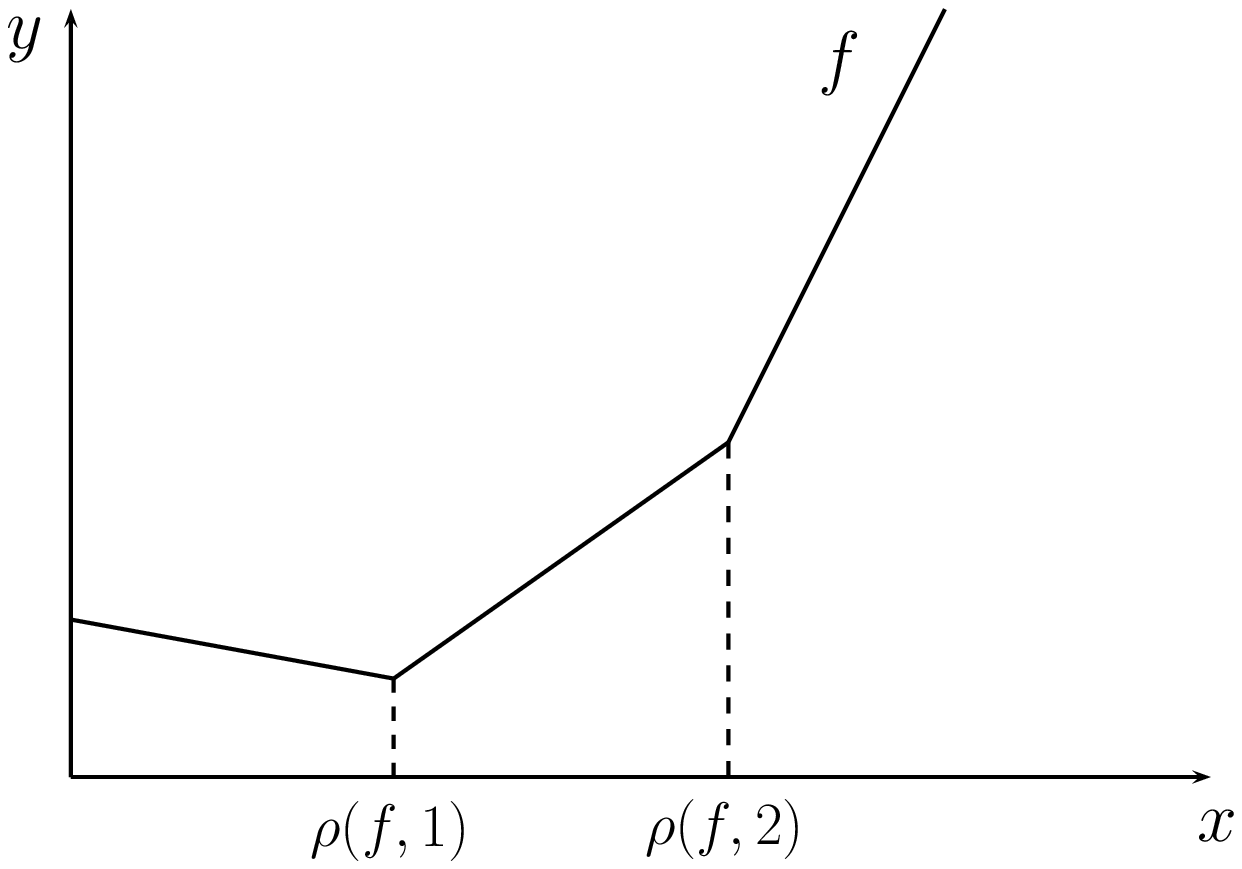} 
    \includegraphics[scale=0.48]{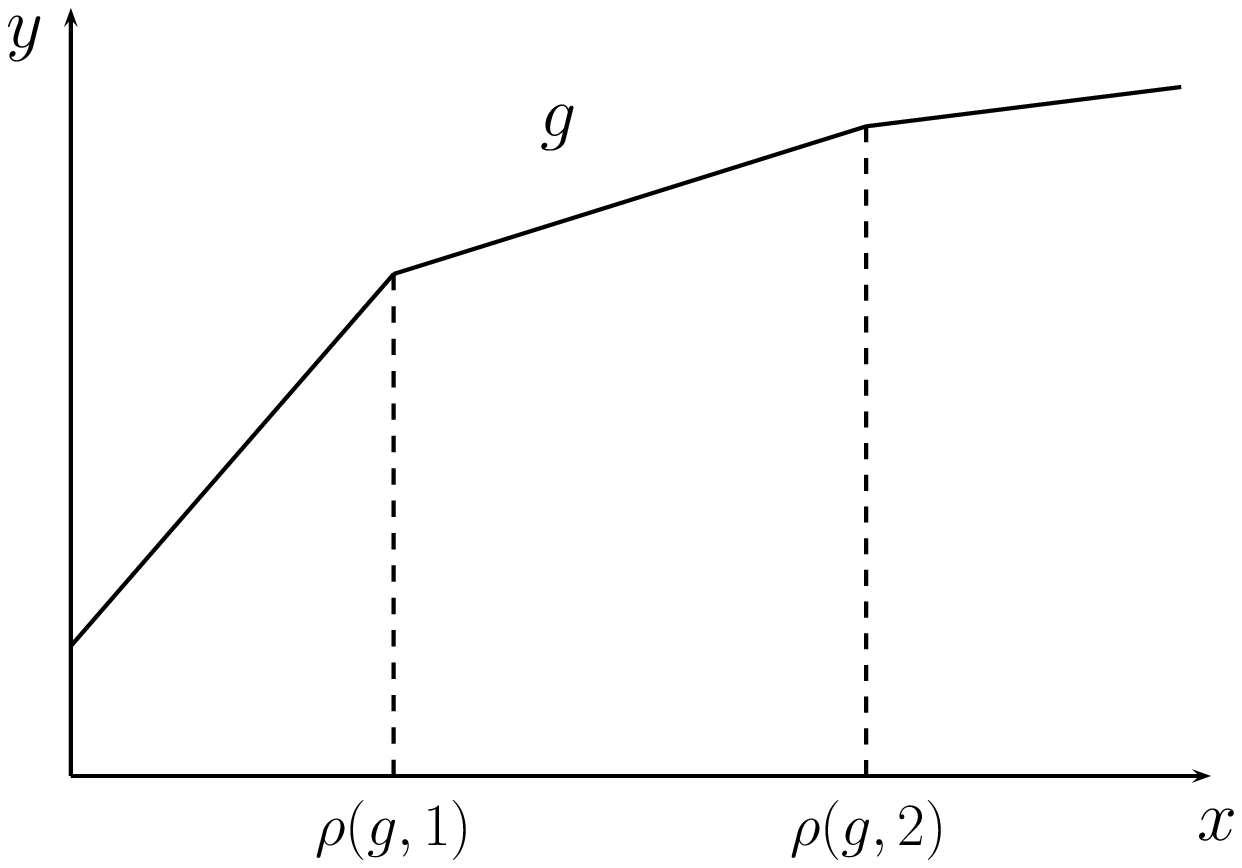}
  \end{center}
  \vspace{-0.3cm}
  \caption{An example of a piecewise linear convex function $f$ (left plot) and a piecewise linear concave function $g$ (right plot) with three pieces.}
  \label{fig:piecewise_linear}
\end{figure}

\paragraph{Mixed Integer Programs with Simple Linear Convex/Concave Transformations.}

The following problem captures our extension of the \MIPlong problem and
will be our central technical tool in the following sections.

\newcommand{\rn}{\ensuremath{(n+r)}}
\newcommand{\mk}{\ensuremath{(m+k)}}
\probDef{\eMIPlong (\eMIP)}
{A collection of $(n+r)m$ piecewise linear convex functions $F = \{f_{i, j} \colon 1 \le i \le \rn, 1 \le j \le m\}$,
 a collection of $(n+r)m$ piecewise linear concave functions $G = \{g_{i, j} \colon 1 \le i \le \rn, 1 \le j \le m\}$,
 and a vector $b \in \mathds{Z}^{m}$.}
{Is there a vector~$x$ such that
\begin{alignat*}{2}
 \sum_{i=1}^{n+r} f_{i, 1}(x_i) &\leq  \sum_{i=1}^{n+r} g_{i, 1}(x_i) &&+ b_1, \\
 \sum_{i=1}^{n+r} f_{i, 2}(x_i) &\leq  \sum_{i=1}^{n+r} g_{i, 2}(x_i) &&+ b_2, \\
 &~~\vdots && \\
 \sum_{i=1}^{n+r} f_{i, m}(x_i) &\leq  \sum_{i=1}^{n+r} g_{i, m}(x_i) &&+ b_{m},\\
 x_i &\in \mathds{N}    && \text{for }1 \le i \le n, \\
 x_i &\in \mathds{R^+}  && \text{for }n+1 \le i \le n+r?
\end{alignat*}
}

%

\noindent Below we describe how to use \autoref{thm:mip} to solve \eMIP with 
up to polynomial factors the 
same asymptotic time complexity as \MIP.

\begin{theorem}\label{thm:general_theorem}
 There is an algorithm that solves
 \eMIPlong  in $O(n^{2.5n +o(n)}\cdot (|I|+\pmax)^{O(1)})$ time, where
 $n$~is the number of integer variables,
 $\pmax$~is the maximum number of pieces per function,
 and $|I|$~is the number of bits encoding the input.
\end{theorem}

\begin{proof}
  To prove the theorem, we will reduce \eMIP to \MIP, by replacing
  each non-linear constraint with a polynomial number of linear ones,
  using polynomially many additional real-valued variables (but
  without changing the number of integer ones).  This will allow us to
  invoke \autoref{thm:mip} to solve the resulting program.

  Note that in the \eMIP problem (as in the \MIP one), we consider the
  canonical form, where all variables are nonnegative.  Hence, we can
  assume that the zeroth piece of each function~$f_{i,j}$ and each
  function~$g_{i,j}$ includes point~$0$ (pieces covering only negative
  points are irrelevant).  Furthermore, by appropriately setting the
  $b_j$~coefficients, we can also assume that for each $i$ and $j$ we
  have $f_{i,j}(0)=0$ and $g_{i,j}(0)=0$.

  Our reduction starts with the input \eMIP instance and successively
  transforms it into an ordinary \MIP instance.  We keep all integer
  variables $x_1, \ldots, x_n$ and all real-valued variables $x_{n+1},
  \ldots, x_{n+r}$ of the original \eMIP instance, but we introduce
  additional real-valued variables and we replace non-linear
  constraints with a number of linear ones.

  \paragraph{Replacing a Non-Linear Constraint.}
  Let $j$ be an integer such that the non-linear constraint
  \begin{align}\label{eq:really_original}
      \sum_{i=1}^{n+r} f_{i, j}(x_i) \leq  \sum_{i=1}^{n+r} g_{i, j}(x_i) + b_{j}
  \end{align}
  has not yet been replaced with linear ones. We remove it from the
  program and in its place we include constraint (where $w_{i,j}$ and
  $u_{i,j}$ are new real-valued variables):
  \begin{align}\label{eq:original}
    \sum_{i=1}^{n+r} w_{i,j} &\leq \sum_{i=1}^{n+r} u_{i, j} + b_{j} \text{.}
  \end{align}
  We will
  include additional constraints so that variables $w_{i,j}$ will
  upper-bound values $f_{i,j}(x_i)$ and variables $u_{i,j}$ will
  lower-bound values $g_{i,j}(x_i)$ (and, indeed, we will be able to
  assume that these variables have exactly the values $f_{i,j}(x_i)$
  and $g_{i,j}(x_i)$, respectively).

  \paragraph{Additional Variables.}

  To ensure that variables $w_{i,j}$ and $u_{i,j}$ have correct values
  (in the feasible solution for our program), we need additional
  variables. For each $i \in [n+r]$, we introduce $\pieces(g_{i,j})$
  real-valued variables $y_{i, j, 1}, \ldots, y_{i,j,
    \pieces(g_{i,j})}$ and $\pieces(f_{i,j})$ real-valued variables
  $z_{i, j, 1}, \ldots, z_{i,j, \pieces(f_{i,j})}$.  We will ensure
  that if there is a feasible solution to our program then there is
  one where $y_{i,j,\ell} = \max(0, x_i - \piece(g_{i,j}, \ell))$ and
  $z_{i,j,\ell} = \max(0, x_i - \piece(f_{i,j}, \ell))$.
  In other words, for each $\ell$, the variable $y_{i,j,\ell}$
  (resp. $z_{i,j,\ell}$) measures how far the variable $x_i$ is beyond
  the beginning of the $\ell$-th piece of function $g_{i,j}$ (resp. of
  function $f_{i,j}$).



  \paragraph{Constraining Variables $\boldsymbol{w_{i,j}}$ (Convex
    Case).}
  First, for each variable $z_{i, j, \ell}$ we introduce two
  constraints:
  \begin{align}\label{eq:constraint_y}
    \begin{split}
      z_{i, j, \ell} &\geq 0,\\
      z_{i, j, \ell} &\geq x_i - \piece(f_{i,j}, \ell).
    \end{split}
  \end{align}
  Second, for each $i \in [n+r]$ we introduce the constraint:
  \begin{align}\label{eq:constraint_z}
    x_i \cdot \der(f_{i,j}, 0) + \sum_{\ell =
      1}^{\mathclap{\quad\pieces(f_{i,j})}} z_{i, j, \ell} \cdot (\der(f_{i,j}, \ell) -
    \der(f_{i,j}, \ell-1)) \leq w_{i, j}.
  \end{align}

  \paragraph{Constraining Variables $\boldsymbol{u_{i,j}}$ (Concave
    Case).}
  Analogously to the convex case, we first introduce two constraints
  for each variable $y_{i, j, \ell}$ (these two constraints are almost
  identical to the convex case):
  \begin{align}\label{eq:constraint_y_concave}
    \begin{split}
      y_{i, j, \ell} &\geq 0,\\
      y_{i, j, \ell} &\geq x_i - \piece(g_{i,j}, \ell).
    \end{split}
  \end{align}
  Second, for each $i \in [n+r]$ we introduce the constraint:
  \begin{align}\label{eq:constraint_z_concave}
    u_{i, j} \le x_i \cdot \der(g_{i,j},0) + \sum_{\ell = 1}^{\mathclap{\quad \pieces(g_{i,j})}} y_{i,j,\ell}
    \bigg(\der(g_{i,j},\ell)-\der(g_{i,j}, \ell-1)\bigg).
  \end{align}

  \paragraph{Correctness.}
  We now argue that if there is a feasible solution for our
  transformed \MIP instance, then there also is one for the original
  \eMIP instance. Let us assume some feasible solution for the
  transformed instance and focus on some arbitrary constraint number $j$
  (recall Equations~\eqref{eq:really_original} and~\eqref{eq:original}).

  In order to satisfy Constraint~\eqref{eq:original}, smaller values
  of $w_{i, j}$ are clearly more desirable.  Since each~$w_{i, j}$
  only occurs once on the left-hand side in
  Constraint~\eqref{eq:original} and once on the right-hand side in one
  of the constraints of the from~\eqref{eq:constraint_z}, we can
  assume that each constraint of the form~\eqref{eq:constraint_z} is
  satisfied with equality (given a feasible solution, we can keep on
  decreasing the values $w_{i,j}$ until we hit equalities in these
  constraints).  Further, together with the fact that $f_{i, j}$ is
  convex, and consequently $\der(f_{i,j}, \ell) > \der(f_{i,j},
  \ell-1)$ for each $\ell$, we infer that the values $z_{i, j, \ell}$
  can be as small as possible.  Formally, similarly as above, using
  constraints of the form~\eqref{eq:constraint_y}, we can assume that
  in our feasible solution  for each
  variable~$z_{i, j, \ell}$ it holds that $z_{i, j, \ell} = \max(0,
  x_i - \piece(f_{i,j}, \ell))$.  Consequently, we conclude that there
  is a feasible solution where:
  \begin{align*}
    w_{i, j} = x_i \cdot \der(f_{i,j}, 0) + \sum_{\ell = 1}^{\mathclap{\quad \pieces(f_{i,j})}}
    \max(0, x_i - \piece(f_{i,j}, \ell)) \cdot (\der(f_{i,j}, \ell) -
    \der(f_{i,j}, \ell-1)).
  \end{align*}
  Let us now analyze the value $w_{i,j}$ provided by this equality.
  If $x_i=0$, then we surely have $w_{i, j} = 0 = f_{i,j}(0)$.  Next,
  we analyze how the value of $w_{i,j}$ changes when we increase~$x_i$
  to~$\Delta$.  If $x_i$ does not exceed $\piece(f_{i,j}, 0)$, then
  $w_{i,j}$ has value $\Delta \cdot \der(f_{i,j}, 0)$ and we have
  $w_{i,j} = f(\Delta)$.  If $x_i$ is greater than $\piece(f_{i,j},
  0)$ but smaller than $\piece(f_{i,j}, 1)$, then $w_{i,j}$ has value
  $\Delta \cdot \der(f_{i,j}, 0) + (\Delta - \piece(f_{i,j}, 0)) \cdot
  (\der(f_{i,j}, 1) - \der(f_{i,j}, 0)) = f_{i,j}(\Delta)$. By
  analogous reasoning, we obtain $w_{i,j} = f_{i,j}(x_i)$ for every
  value of $x_i$.

  Analogous reasoning shows that we can also assume that under our
  feasible solution it holds that $u_{i,j} =
  g_{i,j}(x_i)$. Specifically, we note that to satisfy
  Constraint~\eqref{eq:original}, larger values of $u_{i, j}$ are
  clearly more desirable.  Since each~$u_{i, j}$ only occurs once on
  the right-hand side in Constraint~\eqref{eq:original} and once on the
  left-hand side in one constraint of the
  form~\eqref{eq:constraint_z_concave}, we infer that each constraint
  of the form~\eqref{eq:constraint_z_concave} can be satisfied with
  equality.  Further, together with the fact that $g_{i, j}$ is
  concave, and consequently $\der(g_{i,j}, \ell) < \der(g_{i,j},
  \ell-1)$ for each $\ell$, we infer that the values $y_{i, j, \ell}$
  can be as small as possible.  Formally, similarly as above, we infer
  from constraints of the form~\eqref{eq:constraint_y_concave} that we
  can assume that in our feasible solution for each variable~$y_{i, j,
    \ell}$ it holds that $y_{i, j, \ell} = \max(0, x_i -
  \piece(g_{i,j}, \ell))$.  Consequently, we have that:
  \begin{align*}
    u_{i, j} = x_i \cdot \der(g_{i,j}, 0) + \sum_{\ell = 1}^{\mathclap{\quad \pieces(g_{i,j})}}
    \max(0, x_i - \piece(g_{i,j}, \ell)) \cdot
    \bigg(\der(g_{i,j},\ell)-\der(g_{i,j}, \ell-1)\bigg).
  \end{align*}
  Analysis analogous to that for the case of variables $w_{i,j}$ shows
  that, indeed, we have $u_{i,j} = g_{i,j}(x_i)$.

  The above reasoning, together with the fact that constraints of the
  form~\eqref{eq:original} are satisfied and clearly correspond to the
  original constraints in the \eMIP problem, proves that if there is a
  feasible solution for the transformed instance, then there also is
  one for the original instance. 
  From our reasoning it is also apparent that the other implication holds
  (given variables $x_i$, it suffices that for each $i$, $j$, and
  $\ell$ we set $w_{i,j} = f_{i,j}(x_i)$, $u_{i,j} = g_{i,j}(x_i)$,
  $z_{i, j, \ell} = \max(0, x_i - \piece(f_{i,j}, \ell))$, and $y_{i,
    j, \ell} = \max(0, x_i - \piece(g_{i,j}, \ell))$).



  \paragraph{Running Time.}
  The running-time of the algorithm from \autoref{thm:mip} invoked
  on our transformed instance is upper-bounded by
  ${n}^{2.5\cdot{n}+o({n})}\cdot |I^*|^{O(1)}$, where $n$~denotes the
  number of integer variables and $|I^*|$~is the number of bits needed
  to encode our MIP.  Finally, $|I^*|^{O(1)}$ can be upper-bounded by
  $(|I|+\pmax)^{O(1)}$ since we introduced at most
  $O(\pmax\cdot(n+r))$ additional constraints and variables.  This
  completes the proof.
\end{proof}

We conclude this section with two observations regarding the
generality of \autoref{thm:general_theorem}.  First, note that in
this section we used the canonical form of \eMIP, requiring all the
variables to be non-negative.  Yet, as long as we do not actually use
the piecewise linear transformations on a variable~$x_i$ (that is, as
long as for each $j$, functions $f_{i,j}$ and $g_{i,j}$ are linear) we
can use the standard technique of replacing each occurrence of~$x_i$
with $x_i^+ - x_i^-$, where $x_i^+$ and $x_i^-$ are two nonnegative
variables denoting, respectively, the positive and the negative part
of~$x_i$.  In other words, we may allow negative values for each
variable~$x_i$ whose associated functions~$f_{i,j}$ and~$g_{i,j}$ are
simply linear.

Second, note that if a certain function $f_{i,j}$ is applied to an
integer variable $x_i$ and we know that the value of $x_i$ must come
from some set $\{0, \ldots, t_i\}$, then it suffices that $f_{i,j}$ is
convex/concave on integer arguments, and we do not have to worry about
it being piecewise linear. This last condition is vacuously satisfied:
It suffices to consider pieces $[0,1), [1,2), [2,3), \ldots,
[t_i,t_i+1)$, and derivatives $\der(f,0) = f(1)-f(0)$, $\der(f,1) =
f(2)-f(1)$, $\der(f,2) = f(3)-f(2)$, and so on. It turns out that in
many applications we apply convex/concave functions to integer
variables only and this observation comes handy.


\section{Covering and Approval Voting: Showcases of the 
  Technique}\label{sec:showcase}

In this section we demonstrate how to apply
\autoref{thm:general_theorem} to resolve the complexity of a
number of problems related to covering and approval voting.  These are
interesting because the complexity of some of the problems we consider
was open for the last ten years or so.  We also show that
\autoref{thm:general_theorem} is useful in designing an $\fpt$
approximation scheme for the \textsc{Multiset Multicover} problem.
This result is interesting because it is somewhat involved
technically and illustrates mixing of our \eMIP-based approach with
combinatorial arguments.

\subsection{Weighted Multiset Multicover with Small Universe}\label{section_main}

We start by focusing on the complexity of a few generalizations of the
\textsc{Max Cover} problem.
If $A$ is a multiset and $x$ is some element, then we write $A(x)$ to
denote the number of times $x$ occurs in~$A$ (that is, $A(x)$ is $x$'s
multiplicity in~$A$). If $x$ is not a member of~$A$, then~$A(x) = 0$.

\begin{definition}
  In the \textsc{Weighted Multiset Multicover} (WMM) problem we are
  given a family $\calS = \{S_1, \ldots, S_n\}$ of multisets over the
  universe $U = \{x_1, \ldots, x_m\}$, integer weights $w_1, \ldots,
  w_n$ for the multisets, integer covering requirements $r_1, \ldots,
  r_m$ for the elements of the universe, and an integer budget $B$. We
  ask whether there exists a subfamily~$\calS' \subseteq \calS$ of
  multisets from $\calS$ such that:
  \begin{enumerate}
  \item for each $x_i \in U$ it
  holds that $\sum_{S_j \in \calS'}S_j(x_i) \geq r_i$ (that is, each
  element~$x_i$ is covered at least the required number of times), and 
  \item $\sum_{S_j \in
    \calS'}w_j \leq B$ (the budget is not exceeded).
\end{enumerate}
\end{definition}

We will show how the complexity of \textsc{WMM} (parameterized with
respect to the universe size) changes as we keep on adding
restrictions.
First, we observe that a straightforward polynomial-time reduction
from \textsc{Partition} proves that \textsc{WMM} is $\np$-hard even
for the case of a single-element universe.  Clearly, this also means
that the problem is $\paranp$-hard with respect to the number of
elements in the universe as the parameter.

\begin{proposition}\label{prop:wmsmcpara}
  \textsc{WMM} is $\np$-complete even for the case of a single-element
  universe.
\end{proposition}

\begin{proof}
  Membership in $\np$ is clear. We show $\np$-hardness by a reduction
  from the \textsc{Partition} problem.  An instance of
  \textsc{Partition} consists of a sequence of nonnegative integers
  $k_1, \ldots, k_n$. We ask if there is a set $I \subseteq
  [n]$ such that $\sum_{i \in I} k_i = \frac{1}{2} \sum_{i=1}^n k_i = \sum_{i \notin I} k_i$.
  
  We form an instance of $\wmsmc$ as follows.  The universe contains a
  single element $x$ with covering requirement equal to
  $\frac{1}{2}\sum_{i=1}^nk_i$.  For each $i$, $1 \leq i \leq n$,
  there is a single multiset $S_i$ containing $k_i$ occurrences of
  $x$, with weight $k_i$. We set the budget to be
  $\frac{1}{2}\sum_{i=1}^nk_i$.  Clearly, it is possible to cover~$x$
  sufficiently many times if and only if our input instance of
  \textsc{Partition} is a ``yes''-instance.
\end{proof}

Another variant of \textsc{WMM} is \textsc{Multiset Multicover}, where
we assume each set to have unit weight.
By generalizing the proof for \autoref{prop:wmsmcpara},
we show that this problem is 
$\np$-hard already for two-element universes, which again implies
$\paranp$-hardness with respect to the number of elements in the
universe.

\begin{proposition}\label{prop:msmcpara}
  $\msmc$ is $\np$-complete even for universes of size two. 
\end{proposition}

\begin{proof}
  Membership in $\np$ is clear. To show $\np$-hardness, we
  give a reduction from a variant of the \textsc{Subset Sum}
  problem. We are given a sequence $k_1, \ldots, k_{2n}$ of positive
  integers, a target value~$T$, and we ask if there is a set $I
  \subseteq \{1, \ldots, 2n\}$ such that (a) $\sum_{i \in I}k_i = T$,
  and (b) $\|I\| = n$.

  Let $K$ be $\max(k_1, \ldots, k_{2n})$.  We form an instance of
  $\msmc$ that contains two elements, $x_1$ and $x_2$.  For each $i$,
  $1 \leq i \leq 2n$, we form a set $S_i$ that contains $x_1$ with
  multiplicity~$k_i$, and $x_2$ with multiplicity~$nKT-k_i$. We set
  the covering requirement $r_1$ of $x_1$ to be $T$, and the covering
  requirement $r_2$ of $x_2$ to be $n^2KT-T$. We ask if there is a
  multiset multicover of size at most $n$.

  Clearly, if there is a solution for our \textsc{Subset Sum}
  instance, then the sets that correspond to this solution form a multiset
  multicover of $x_1$ and $x_2$. On the contrary, assume that there
  is a collection of at most $n$ sets that form a multiset multicover
  of $x_1$ and~$x_2$.  There must be exactly~$n$ of these
  sets. Otherwise, the sum of their multiplicities for~$x_2$ would be
  smaller than~$n^2KT-T$.  Due to the covering requirement of~$x_1$,
  these sets correspond to the numbers from $\{x_1,\ldots, x_{2n}\}$
  that sum up to at least $T$, and due to covering requirement of
  $x_2$, these sets correspond to numbers that sum up to at most
  $T$. This completes the proof.
\end{proof}

Often we do not need the full flexibility of \textsc{WMM}. For
instance, in the next section we will describe several problems from
computational social choice that can be reduced to more specific
variants of \textsc{WMM}, such as the \textsc{Weighted Set Multicover}
and \textsc{Uniform Multiset Multicover} problems. 
\textsc{Weighted Set Multicover} is a variant of \textsc{WMM} where
each input multiset has elements with multiplicities 0 or~1 (in other
words, the family $\calS$ contains sets without multiplicities, but
the union operation takes multiplicities into account).
\textsc{Uniform Multiset Multicover} is a variant of \textsc{Multiset
  Multicover} (and, thus, of \textsc{WMM}), where for each multiset
$S_i$ in the input instance there is a number $t_i$ such that for each
element $x$
we have $S_i(x) \in \{0,t_i\}$ (in other words, elements within a
single multiset have the same multiplicities). 

As the first application of our new framework, we show that
\textsc{Weighted Set Multicover} is fixed-parameter tractable when
parameterized by the universe size.  Notably, we only use convex
constraints in the construction of the \eMIP instance.

\begin{theorem}\label{thm:wsmc-fpt}
  \textsc{Weighted Set Multicover} is fixed-parameter tractable 
when parameterized by
  the universe size.
\end{theorem}
\begin{proof}
  Consider an instance of \textsc{Weighted Set Multicover} with
  universe $U = \{x_1, \ldots, \allowbreak x_m\}$, family $\calS =
  \{S_1, \ldots, S_n\}$ of subsets, weights $w_1, \ldots, w_n$ for the
  sets, covering requirements $r_1, \ldots, r_m$ for the elements, and
  budget $B$. Our algorithm proceeds by solving an appropriate \eMIP
  instance.

  First, we form a family $U_1, \ldots, U_{2^m}$ of all subsets 
  of~$U$. For each~$i$, $1 \leq i \leq 2^m$, let $\calS(U_i):=\{S_j \in \calS \mid S_j=U_i \}$.
  For each~$i$ and $j$, $1 \leq i \leq 2^m$ we define a convex
  function $f_i$ so that for each integer $j$, $1 \leq j \leq
  |\calS(U_i)|$, $f_i(j)$ is the sum of the $j$ lowest weights of the
  sets from $\calS(U_i)$.
  We have $2^m$ integer variables $z_i$, $1 \leq i \leq
  2^m$. Intuitively, these variables describe how many sets we take
  from each type (i.e., how many sets we take from each family
  $\calS(U_i)$).
  We introduce the following constraints: For each $i$, $1 \leq i \leq
  2^m$, we have constraints $z_i \geq 0$ and $z_i \leq |\calS(U_i)|$.  For
  each element $x_\ell$ of the universe, we also have constraint
  $\sum_{U_i\colon x_\ell \in U_i}z_i \geq r_\ell$.  These constraints
  ensure that the variables $z_i$ describe a possible solution for the
  problem (disregarding the budget). Our last constraint uses
  variables $z_i$ to express the requirement that the solution has
  cost at most $B$:
\begin{align*}
     \textstyle\sum_{i=1}^{2^m}f_i(z_i) \leq B \text{.}
\end{align*}
  Finally, we use \autoref{thm:general_theorem} to get the statement of the theorem.
\end{proof}

As a second application of our new framework, we show that
\textsc{Uniform Multiset Multicover} is fixed-parameter tractable when
parameterized by the universe size.  Notably, we only use concave
constraints in the corresponding \eMIP instance.

\begin{theorem}\label{thm:msmc-fpt}
  \textsc{Uniform Multiset Multicover} is fixed-parameter tractable 
when parameterized 
by the universe size.
\end{theorem}

\begin{proof}
  Consider an instance of \textsc{Uniform Multiset Multicover} with
  universe $U = \{x_1, \ldots, \allowbreak x_m\}$, family $\calS =
  \{S_1, \ldots, S_n\}$ of subsets, covering requirements $r_1,
  \ldots, r_m$ for the elements, and budget $B$.  Our algorithm
  proceeds by solving an appropriate \eMIP instance.

  Similarly as in the proof of \autoref{thm:wsmc-fpt}, we form a
  family $U_1, \ldots, U_{2^m}$ of all the subsets of $U$ (note that
  these, indeed, are subsets and not multisets).  For each $i$, $1
  \leq i \leq 2^m$, let $\calS(U_i)$ be a subfamily of $S$ that contains
  those multisets in which exactly the elements from $U_i$ appear
  (that is, their multiplicities are non-zero).  For each~$i$, $1 \leq
  i \leq 2^m$, we define a concave function $f_i$ so that for
  each~$j$, $1 \leq j \leq |\calS(U_i)|$, $f_i(j)$ denotes the maximum
  sum of multiplicities for each element from~$U_i$ using~$j$
  multisets from~$\calS(U_i)$.  (To compute this function, we simply need
  to sort the multisets from~$\calS(U_i)$ in the order of decreasing
  multiplicities.  Then, $f_i(j)$ is the sum of the multiplicities
  with respect to an arbitrary element from~$U_i$ of the first
  $j$~multisets.)

  We have $2^m$ integer variables $z_i$, $1 \leq i \leq 2^m$.
  Intuitively, the $z_i$~variables describe how many multisets we take from each type.
  Thus, $f_i(z_i)$~describes how much each element from~$U_i$ is covered by taking $z_i$~multisets of type~$U_i$.
  We introduce the following constraints: For each $i$, $1 \leq i \leq
  2^m$, we have constraints $z_i \geq 0$ and $x_i \leq |\calS(U_i)|$.  For
  each element $x_\ell$ of the universe, we also have constraint
  $\sum_{U_i\colon x_\ell \in U_i} f_i(z_i) \geq r_\ell$.  These
  constraints ensure that the variables $z_i$ describe a possible
  solution for the problem (disregarding the budget).  To express the
  requirement that the solution has cost at most~$B$, we add the
  constraint $\sum_{i=1}^{2^m}z_i \leq B$.
  Finally, we use \autoref{thm:general_theorem} to get the statement of the theorem.
\end{proof}

Unfortunately, it is impossible to apply our approach to the more general \textsc{Multiset
  Multicover}; by \autoref{prop:msmcpara},
\textsc{Multiset Multicover}
is already $\np$-hard for two-element universes. It is,
however, possible to obtain a certain
form of an $\fpt$ approximation scheme.

\begin{definition}
  Let $\epsilon$ be a real number, $\epsilon > 0$.  We say that
  algorithm $\mathcal{A}$ is an \emph{$\epsilon$-almost-cover
    algorithm} for \textsc{Multiset Multicover} if, given an input
  instance $I$ with universe $U = \{x_1, \ldots, x_m\}$ and covering
  requirements $r_1, \ldots, r_m$, it outputs a solution that covers
  each element~$x_i$ with multiplicity~$r_i'$ such that $\sum_i
  \max(0, r_i - r_i') < \epsilon \sum_i r_i$.
\end{definition}

In other words, on the average an $\epsilon$-almost-cover algorithm
can miss each element of the universe by an $\epsilon$-fraction of its
covering requirement.  For the case where we really need to cover all
the elements perfectly, we might first run an $\epsilon$-almost-cover
algorithm and then complement its solution, for example, in some greedy
way, since the remaining instance might be much easier to solve.

The key idea regarding computing an $\epsilon$-almost-cover is that it
suffices to replace each input multiset by several sub-multisets, each
with a particular ``precision level,'' so that multiplicities of the
elements in each sub-multiset are of a similar order of magnitude.
The full argument, however, forms the most technical
part of our paper.


\begin{theorem}\label{thm:msmc-fpt-as}
  For every rational $\epsilon > 0$, there is an $\fpt$-time
  $\epsilon$-almost-cover algorithm for \textsc{Multiset Multicover}
  parameterized by the universe size.
\end{theorem}

\begin{proof}
  We describe our $\epsilon$-almost-cover algorithm for
  \textsc{Multiset Multicover} and argue its correctness.  We consider
  an instance $I$ of \textsc{Multiset Multicover} with a family $\calS
  = \{S_1, \ldots, S_n\}$ of multisets over the universe $U = \{x_1,
  \ldots, x_m\}$, where the covering requirements for the elements of
  the universe are $r_1, \ldots, r_m$.  We associate each set $S$ from
  the family $\calS$ with the vector $v_{S} = \langle S(x_1), S(x_2),
  \ldots, S(x_m) \rangle$ of element multiplicities.

  Let $\epsilon > 0$ be the desired approximation ratio.
  We fix $Z = \lceil \frac{4m}{\epsilon} \rceil$ and $Y =
  Z + \lceil \frac{4Zm^3}{\epsilon} \rceil$.  Note that
  $\frac{m}{Z} \leq \frac{\epsilon}{4}$ and $\frac{Zm^3}{Y-Z} \leq
  \frac{\epsilon}{4}$.
  Let $X = \left(\frac{2Y^m}{\epsilon} + 1\right)^m$ and let $V_1,
  \ldots, V_X$ be a sequence of all $m$-dimensional vectors whose
  entries come from the $\left(\frac{2Y^m}{\epsilon} +
    1\right)$-element set $\left\{0, \frac{\epsilon}{2}, \epsilon,
  \frac{3\epsilon}{2}, 2\epsilon, \dots, Y^m \right\}$. For each $j$, $1
  \leq j \leq X$, we write $V_j = \langle V_j(x_1), V_j(x_2), \dots,
  V_{j}(x_m) \rangle$. Intuitively, these vectors describe some subset
  of ``shapes'' of all possible multisets---interpreted as vectors of
  multiplicities---over our $m$-element universe.  For each number
  $\beta$, we write $\beta V_i$ to mean the vector $\langle \lfloor
  \beta V_{i, 1} \rfloor, \lfloor \beta V_{i, 2} \rfloor, \dots,
  \lfloor \beta V_{i, m} \rfloor\rangle$.

  Vectors of the form~$\beta V_i$ are approximations of
  those multisets for which the positive multiplicities of the
  elements do not differ
  too much (formally, for those multisets for which the positive
  multiplicities differ by at most a factor of $Y^m$). Indeed, for
  each such multiset~$S$, we can find a value $\beta$ and a vector
  $V_j$ such that for each element $x_i$ it holds that $S(x_i) \geq
  \beta V_j(x_i) \geq \left(1 -
    \frac{\epsilon}{2}\right)S(x_i)$. However, this way we cannot
  easily approximate those sets for which multiplicities differ by
  large factors. For example, consider a set $S$ represented through
  the vector $\langle 0, \dots, 0, 1, Q\rangle$, where $Q \gg Y^m$, in
  particular where $Q \gg \frac{2Y^m}{\epsilon}$. For each value
  $\beta$ and each vector $V_j$, the vector $\beta V_{j}$ will either
  be inaccurate with respect to the multiplicity of element $x_{m-1}$
  or with respect to the multiplicity of element $x_m$ (or with
  respect to both these multiplicities).

  The main step of our algorithm is to modify the instance $I$ so that we
  replace each multiset~$S$ from the family $\calS$ with a sequence of
  vectors of the form $\beta V_j$ that altogether add to at most the
  multiset~$S$ (each such sequence can contain
  multiple vectors of different ``shapes'' $V_j$ and of different
  scaling factors $\beta$). The goal is to obtain an instance that on
  the one hand consists of ``nicely-structured'' sets (vectors) only,
  and on the other hand has the following property: If in the initial
  instance $I$ there exist $K$ sets that cover elements $x_1, \dots,
  x_m$ with multiplicities $r_1, \dots, r_m$, then in the new
  instance there exist $K$ sets that cover elements $x_1, \dots, x_m$
  with multiplicities $r_1', \dots, r_m'$,
  such that $\sum_i \max(0, r_i - r_i') <  \epsilon \sum_i r_i$.
  We refer to this as the \emph{almost-cover approximation property}.

  The procedure for replacing a given set $S$ is presented as
  \autoref{alg:msmc-fpt-as-transformation}. This algorithm calls
  the \texttt{Emit} function with arguments $(\beta, V)$ for each
  vector $\beta V$ that it wants to output ($V$ is always one of the
  vectors $V_1, \ldots, V_X$). The emitted sets replace the set $S$
  from the input.  Below we show that if we apply \autoref{alg:msmc-fpt-as-transformation}
  to each set from $\calS$,
  then the resulting instance $I'$ has our almost-cover approximation
  property.

\newcommand{\mycommentfont}[1]{{\emph{#1}}}

\SetCommentSty{mycommentfont}
\LinesNumbered
\begin{algorithm}[t!]
   \SetKwInput{KwNotation}{Notation}
   \SetKwFunction{Main}{Main}
   \SetKwFunction{sort}{sort}
   \SetKwFunction{Emit}{Emit}
   \SetKwFunction{RoundAndEmit}{Round\_And\_Emit}
   \SetKwFunction{MainRec}{Main\_Rec}
   \SetKwBlock{Block}
   \SetAlCapFnt{\footnotesize}
   \caption{\label{alg:msmc-fpt-as-transformation}The transformation algorithm used in the proof of
     \autoref{thm:msmc-fpt-as}. The algorithm replaces a given set $S$
     with a sequence of vectors of the form $\beta V_j$.}
   \Main{$S$}:
   \Block{
      multip $\leftarrow \langle (1, S(x_1)), (2, S(x_2)), \dots, (m, S(x_m)) \rangle$\;
      sorted $\leftarrow$ \sort{$\mathrm{multip}$} \nllabel{algline::line1} \tcp*{sort in ascending order of multiplicities}
      $i \leftarrow 0$ \;
      \tcp{\quad $\mathit{sorted[i].first}$ refers to the $i$'th item's number}
      \tcp{\quad $\mathit{sorted[i].second}$ refers to its multiplicity}
      \While{$\mathrm{sorted[}i\mathrm{].second} = 0$} {
         $i \leftarrow i + 1$ \;
      }
      \MainRec{$i$, $\mathrm{sorted}$} \;
   }
   \hspace{5mm} \\
   \MainRec{$i$, $\mathrm{multip}$}:
   \Block{
      $V \leftarrow \langle 0, 0, \dots, 0 \rangle$ (vector of $m$ zeros). \;
      $\beta \leftarrow$ multip[$i$].second \;
      $V$[multip[$i$].first] $\leftarrow$ 1 \;
      $i \leftarrow i + 1$ \;
      \While{$i \leq m$}{
         \eIf{$\mathrm{multip[}i\mathrm{].second} < Y \cdot \mathrm{multip[}i\mathrm{-1].second}$} {
            $V$[multip[$i$].first] $\leftarrow \frac{\mathrm{multip[}i\mathrm{].second}}{\beta}$ \nllabel{algline::line2}\;
            $i \leftarrow i + 1$ \;
         }
         {
            \For{$j \leftarrow i$ \KwTo $m$ \nllabel{algline::line3}} {
               $V$[multip[$j$].first] $\leftarrow$ $\frac{Z \cdot \mathrm{multip[}i\mathrm{-1].second}}{\beta}$ \nllabel{algline::line35} \;
            }
            \RoundAndEmit{$\beta$, $V$} \nllabel{algline::line4}\;
            \For{$j \leftarrow 1$ \KwTo $m$ \nllabel{algline::line5}} {
               $\mathrm{multip[}j\mathrm{].second} \leftarrow \mathrm{multip[}j\mathrm{].second} - \beta V[\mathrm{multip[}j\mathrm{].first}]$ \nllabel{algline::line6} \;
            }
            \MainRec{$i$, $\mathrm{multip}$} \nllabel{algline::line7} \;
            \Return{}
         }
      }
      \RoundAndEmit{$\beta$, $V$}\;
   }
   \hspace{5mm} \\
   \RoundAndEmit{$\beta$, $V$}:
   \Block{
      \For{$\ell \leftarrow 1$ \KwTo $m$ \nllabel{algline::line8}} {
         $V[\ell] \leftarrow$ $\lfloor \frac{2V[\ell]}{\epsilon} \rfloor / \frac{\epsilon}{2}$\;
      }
      \Emit{$\beta$, $V$}\;
   }
\end{algorithm}


Let us consider how \autoref{alg:msmc-fpt-as-transformation}
proceeds on a given set $S$.  For the sake of clarity, let us assume
there is no rounding performed by
\autoref{alg:msmc-fpt-as-transformation} in function
\texttt{Round\_And\_Emit} (the loop in \autoref{algline::line8}). We
will come back to this issue later.

The algorithm considers the elements of the universe---indexed by
variable $i$ throughout the algorithm---in the order given by the
vector ``sorted'' (formed in \autoref{algline::line1} of
\autoref{alg:msmc-fpt-as-transformation}). Let $\prec$ be the
order in which \autoref{alg:msmc-fpt-as-transformation}
considers the elements (so $x_{i'} \prec x_{i''}$ means that $x_{i'}$
is considered before $x_{i''}$), and let $x'_1, \ldots, x'_m$ be the
elements from the universe renamed so that $x'_1 \prec x'_2 \prec
\cdots \prec x'_m$.  Let $r$ be the number of sets that
\autoref{alg:msmc-fpt-as-transformation} emits on our input set
$S$ and let these sets be $S_1, S_2, \dots, S_r$.  (This is depicted
on \autoref{fig:msmc-fpt-as-transformation}, where for the sake of
the example we take $m = 6$ and $r=3$.)

\setcounter{topnumber}{1}
\begin{figure}[t!]
  \begin{center}
    \includegraphics[scale=1.0]{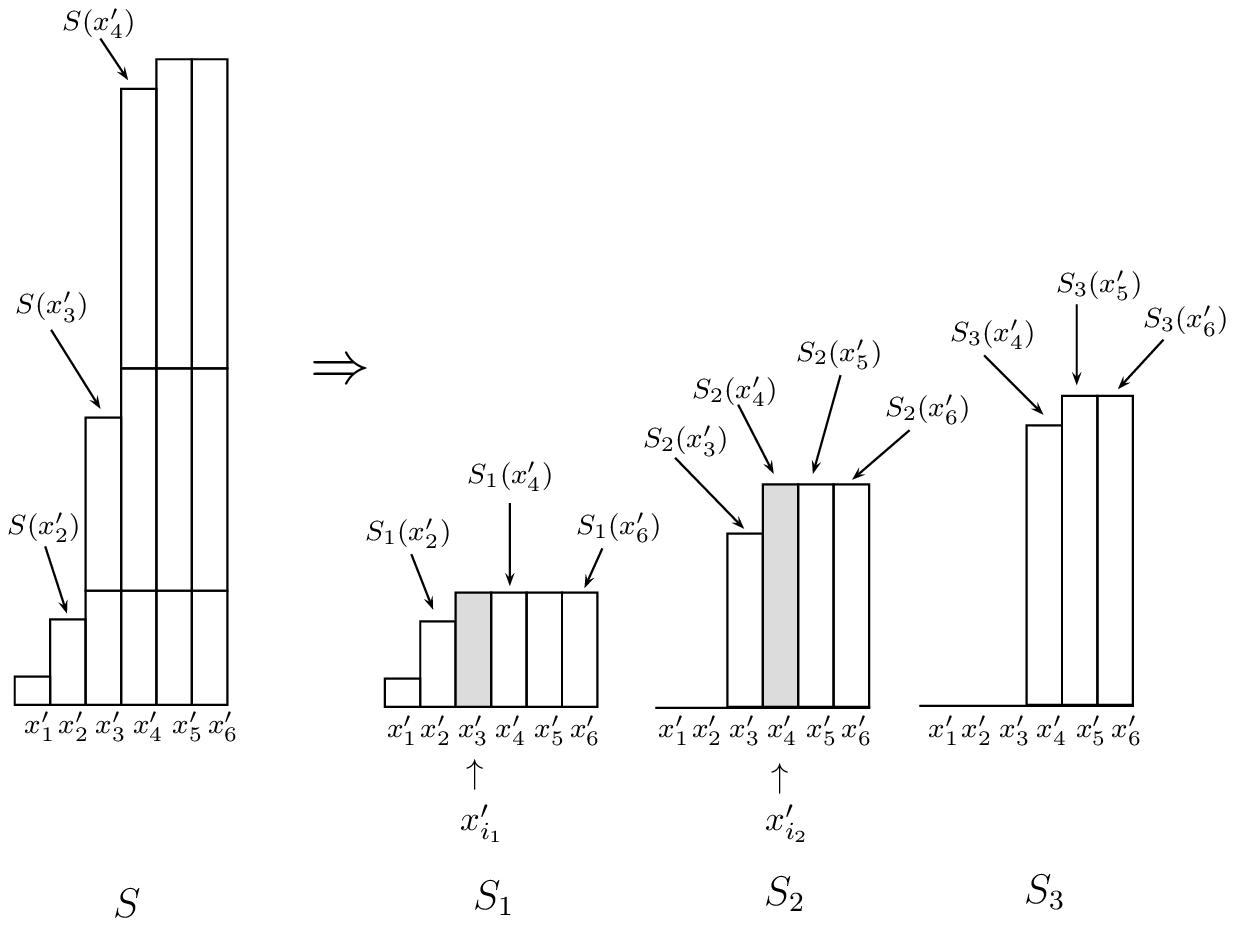}
  \end{center}
  \caption{An example for
    \autoref{alg:msmc-fpt-as-transformation}: The algorithm
    replaces $S$ with sets $S_1$, $S_2$, and $S_3$.}
  \label{fig:msmc-fpt-as-transformation}
\end{figure}

Consider the situation where the algorithm emits the $k$'th set, $S_k$,
and let $i_k$ be the value of variable $i$ right before the call to
\texttt{Round\_And\_Emit} that caused $S_k$ to be emitted.
Note that each element $x$ from the universe such that $x'_{i_k} \prec
x$ has the same multiplicity in~$S_k$ as element~$x'_{i_k}$
(lines~\ref{algline::line3} and~\ref{algline::line35} of
\autoref{alg:msmc-fpt-as-transformation}).  Let $t_k =
\sum_{j}S_k(x'_j)$ be the sum of the multiplicities of the elements
from~$S_k$.
%
We make the following observations:
\begin{description}
\item[Observation 1:] $S_k(x'_{i_k}) = Z \cdot S_k(x'_{i_{k}-1})$.
\item[Observation 2:] It holds that $S_{k+1}(x'_{i_k}) \geq Y \cdot
  S_{k}(x'_{i_k-1}) - Z \cdot S_{k}(x'_{i_k-1}) =
  (Y-Z)S_{k}(x'_{i_k-1}) = \frac{(Y-Z)}{Z}S_k(x'_{i_k})$.
\item[Observation 3:] We have that $S_{k+1}(x'_{i_k}) \geq
  \frac{(Y-Z)}{Z}S_k(x'_{i_k}) \geq \frac{(Y-Z)}{Zm}t_k$. Further, we
  have that $ S_{k+1}(x'_{i_{k+1}}) \geq S_{k+1}(x'_{i_k}) \geq
  \frac{(Y-Z)}{Zm^2} \sum_{j \leq k} t_j$. (To see why the last
  inequality holds, note that $k \leq m$.)

\item[Observation 4:] For $i < i_k$ it holds that $\sum_{q \leq k} S_q(x'_i) = S(x'_i)$.

\end{description}

\newcommand{\opt}{{\mathrm{opt}}}

Now let us consider some solution for instance $I$ that consists of
$K$ sets, $\mathcal{S^\opt} = \{S^\opt_1, S^\opt_2, \dots, S^\opt_K\}
\subseteq \calS$. These sets, altogether, cover all elements from
the universe with required multiplicities. That is, it holds that for
each $i$ we have $\sum_{S \in \mathcal{S^\opt}} S(x_i) \geq r_i$.  For
each set $S \in \calS^\opt$ and for each element $x_i$ from the
universe, we pick an arbitrary number $y_{S,i}$ so that altogether the
following conditions hold:
\begin{enumerate}
\item For every set $S \in \calS^\opt$ and every $x_i$, $y_{S, i} \leq  S(x_i)$.
\item For every $x_i$, $\sum_{S \in \mathcal{S^\opt}} y_{S, i} = r_i$.
\end{enumerate}
Intuitively, for a given set~$S$, the values $y_{S, 1}, y_{S, 2},
\dots, y_{S, m}$ describe the multiplicities of the elements from~$S$
that are \emph{actually used} to cover the elements. Based on these
numbers, we will show how to replace each set from $\calS^\opt$ with
one of the sets emitted for it, so that the resulting family of sets
has the almost-cover approximation property.

\newcommand{\argmax}{\operatornamewithlimits{argmax}}

Consider a set $S \in \calS^\opt$ for which
\autoref{alg:msmc-fpt-as-transformation} emits $r$ sets, $S_1,
S_2, \dots, S_r$. As in the discussion of
\autoref{alg:msmc-fpt-as-transformation}, let $x'_1, \ldots,
x'_m$ be the elements from the universe in which
\autoref{alg:msmc-fpt-as-transformation} considers them (when
emitting sets for $S$). We write $y'_{S,i}$ to mean the value
$y_{S,j}$ such that $x_j = x'_i$.  Let $\mathcal{R} = \{S_1, S_2,
\dots, S_r\}$, let $i_{\max} = \argmax_i y'_{S, i}$, and let
$S_{\mathrm{repl}}$ be the set from $\mathcal{R}$
defined in the following way:

\begin{figure}[t!]
  \begin{center}
    \hspace{-2cm}\includegraphics[scale=1.1]{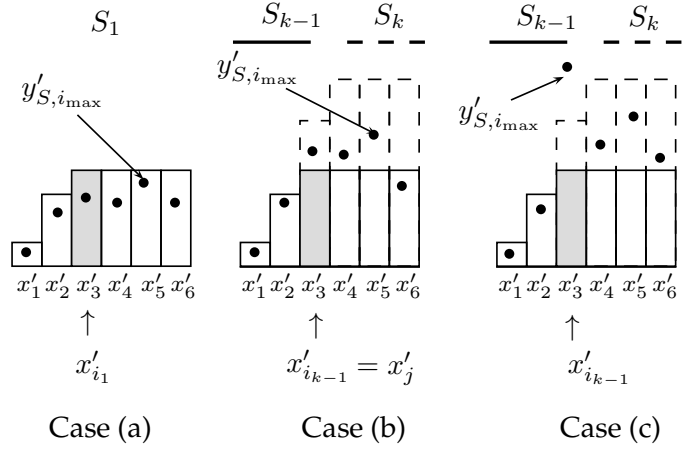}
  \end{center}
  \caption{The cases in the proof of
    \autoref{thm:msmc-fpt-as}. The bullets represent values $y'_{S, 1},
    \dots, y'_{S, m}$.}
  \label{fig:msmc-fpt-as-transformation2}
\end{figure}

\begin{enumerate}
\item If for every set $S_k \in \mathcal{R}$ we have $ 
  S_k(x'_{i_{\max}}) < y'_{S, i_{\max}}$, then $S_{\mathrm{repl}}$ is the
  set $S_k \in \mathcal{R}$ with the greatest value $S_k(x'_{i_{\max}})$ 
  (the set that covers element $x'_{i_{\max}}$ with the greatest
  multiplicity). This is the case denoted as ``Case (c)'' in
  \autoref{fig:msmc-fpt-as-transformation2}.
\item Otherwise $S_{\mathrm{repl}}$ is the set $S_k \in \mathcal{R}$
  that has the lowest value $S_k(x'_{i_{\max}})$, yet no-lower than
  $y'_{S, i_{\max}}$.
  This is the case denoted as either ``Case (a)'' or ``Case (b)'' in
  \autoref{fig:msmc-fpt-as-transformation2}.
\end{enumerate}
We now show that $S_{\mathrm{repl}}$ is a good candidate for replacing
$S$, that is, that 
$$\sum_i \max(0, y'_{S, i} - S_{\mathrm{repl}}(x'_i)) < \epsilon \sum_i y'_{S, i}.$$
To this end, we consider the three cases depicted in
\autoref{fig:msmc-fpt-as-transformation2}:
\begin{description}
\item[Case (a)] It holds that $y'_{S, i_{\max}} <
  S_1(x'_{i_{\max}})$ (that is, $S_1$ already covers the most demanding
  element of the universe to the same extent as $S$ does).  This means
  that we have $\sum_{\ell} \max(0, y'_{S, \ell} - S_1(x'_\ell)) = 0$.
  By the criterion for choosing set $S_{\mathrm{repl}}$, we have that
  $S_{\mathrm{repl}} = S_1$.

\item[Case (b)] There exist sets $S_{k-1}, S_k \in \mathcal{R}$ such
  that $S_k(x'_{i_{\max}}) \geq y'_{S, i_{\max}} >
  S_{k-1}(x'_{i_{\max}})$ (and thus, $S_{\mathrm{repl}} = S_k$). Let
  $x'_j = x'_{i_{k-1}}$ (recall from the discussion of
  \autoref{alg:msmc-fpt-as-transformation} that $i_{k-1}$ is the
  index of the universe element which caused emitting $S_{k-1}$). Let
  us consider two subcases:
  \begin{enumerate}
  \item[(i)] $y'_{S, i_{\max}} \leq S_k(x'_j)$: We first note that for
    each $i \geq j$ it holds that $y'_{S, i} \leq
    S_{k}(x'_i)$. Further, for each $i < j$, we have $y'_{S, i} \leq
    \sum_{\ell \leq k-1} S_\ell(x'_i)$ (this follows from
    Observation~4 and the fact that $y'_{S,i} \leq S(x'_i)$). Based on
    this inequality, we get:
    \begin{align*}
      \sum_{i < j}y'_{S, i} & \leq \sum_{i < j}\sum_{\ell \leq k-1} S_\ell(x'_i)
      \leq \sum_{\ell \leq k-2} t_{\ell} + \sum_{i < j} S_{k-1}(x'_i)  \\
      &\leq \frac{Zm^2}{(Y - Z)} S_{k-1}(x'_j) + \frac{m}{Z} S_{k-1}(x'_j)   \text{\quad\quad\quad\quad(Observations 3 and 1)} \\
      &\leq \frac{\epsilon}{2}S_{k-1}(x'_j) \leq \frac{\epsilon}{2} y_{S,
        i_{\max}} \textrm{.} &
    \end{align*}
    In consequence, it holds that
      $\sum_{\ell} \max(0, y'_{S, \ell} - S_k(x'_\ell)) <
      \frac{\epsilon}{2} \sum_\ell y'_{S, \ell}$.\medskip

    \item[(ii)] $y'_{S, i_{\max}} > S_k(x'_j)$: 
      For $\ell \geq i_k$,
      we have $S_k(x'_\ell) \geq y'_{S, \ell}$ (this follows from the
      fact that $S_k(x'_{i_{\max}}) \geq y'_{S,i_{\max}}$ and the
      definition of $i_k$). For $\ell < i_k$, by Observation 4 and the
      fact that $y'_{S,\ell} \leq S(x'_\ell)$, we have $y'_{S, \ell}
      - S_k(x'_\ell) \leq \sum_{q \leq k-1} S_q(x'_\ell)$. Thus we
      get:
    \begin{align*}
      \sum_{\ell} \max(0,& y'_{S, \ell} - S_k(x'_\ell)) 
      \leq \sum_{\ell < i_k} \max(0, y'_{S, \ell} - S_k(x'_\ell)) \\
      &\leq \sum_{\ell < i_k} \sum_{q \leq k-1} S_q(x'_\ell) 
       \leq \sum_{q \leq k-1} t_{q} \\
      &\leq \frac{Zm^2}{(Y - Z)} S_k(x'_j) & \text{(Observation 3)} \\
      &\leq \frac{Zm^2}{(Y - Z)} y'_{S, i_{\max}} \leq
      \frac{\epsilon}{2} \sum_\ell y'_{S, \ell} \textrm{.} &
    \end{align*}
  \end{enumerate}
\item[Case (c)] Every set $S_{k} \in \mathcal{R}$ has $S_k(x'_{i_{\max}}) <
  y'_{S, i_{\max}}$. 

  By the choice of $S_k$ (the set
  from $\calR$ that has highest multiplicity of $x'_{i_{\max}}$), we
  infer that $\sum_{q \leq k} S_q(x'_{i_{\max}}) =
  S(x'_{i_{\max}})$. 
  Also, for every $\ell < i_{\max}$, we have $\sum_{q \leq k}
  S_q(x'_\ell) = S(x'_\ell)$.
  Consequently, for every $\ell \leq i_{\max}$ we have $$y_{S, \ell} -
  S_k(x'_\ell) \leq \sum_{q \leq k-1} S_{q}(x'_{\ell}) \leq \sum_{q \leq k-1} S_{q}(x'_{i_{\max}}).$$
  Further, for every $\ell > i_{\max}$, we have $$y'_{S, \ell} -
  S_k(x'_\ell) \leq y'_{S, i_{\max}} - S_{k}(x'_{i_{\max}}) \leq
  S(x'_{i_{\max}}) - S_{k}(x'_{i_{\max}}) \leq \sum_{q \leq k-1}
  S_{q}(x'_{i_{\max}}).$$ Based on these observations, we get
  the following:
  \begin{align*}
    \sum_{\ell} \max(0, y_{S, \ell} - S_k(x'_\ell)) &\leq m \sum_{q \leq k-1} S_q(x'_{i_{\max}}) \leq m \sum_{q \leq k-1} t_{q} &\\
    &\leq m \frac{Zm^2}{(Y - Z)} S_{k}(x'_{i_k})  \text{\quad\quad\quad\quad\quad\quad\quad(Observation 3)} \\
    &\leq \frac{Zm^3}{(Y - Z)} y'_{S, i_{\max}} \leq
    \frac{\epsilon}{2} y'_{S, i_{\max}} \leq \frac{\epsilon}{2}
    \sum_\ell y'_{S, \ell} \textrm{.} &
  \end{align*}
  Thus we obtain the desired bound.
\end{description}

The above case analysis almost shows that we indeed have the
almost-cover approximation property. It remains to consider the issue
of rounding (\autoref{algline::line8} of
\autoref{alg:msmc-fpt-as-transformation}).  This rounding
introduces inaccuracy that is bounded by factor
$\frac{\epsilon}{2}$ and thus, indeed, we do have the almost-cover
approximation property.

Now, given the new instance $I'$, it suffices to find a solution for
$I'$ that satisfies the desired approximation guarantee (that is, a
collection $\calS'$ of at most $K$ sets that form an
$\epsilon$-almost-cover).  It is possible to do so using our technique from \autoref{sec:general_technique}.

Let us recall that the new instance consists of the sets that are of
the form $\beta V_j$ (recall the discussion at the beginning of the
proof). For each vector $V_j$, $1 \leq j \leq X$, we introduce an
integer variable $v_j$, which, intuitively, gives the number of sets
with shape $V_j$ taken into the solution. Further, for each $v_j$, $1
\leq j \leq X$, and each $x_i$, $1 \leq i \leq m$, we introduce a
concave function $f_{i, j}$, so that $f_{i,j}(v_j)$ is the maximum
multiplicity with which element $x_i$ is covered by some $v_j$
``best'' sets of the shape $V_j$ (these are the $v_j$ sets which have
been emitted with the highest values of $\beta$).  Finally, we
introduce variables $\miss_1, \dots, \miss_m$, responsible for
measuring the inaccuracy levels (in other words, $\miss_i$ gives the
missing multiplicity for element~$x_i$).  The constraints for our
mixed integer linear program are given below:
\begin{enumerate}
\item $\sum_{j=1}^{X} v_j \leq B$.
\item For each $j$, $1 \leq j \leq X$: $v_j \geq 0$.
\item $\sum_{i=1}^{m} \miss_i \leq \epsilon \sum_i r_i$.
\item For each $i$, $1 \leq i \leq m$: $\sum_{j=1}^{X} f_{i,j}(v_j) \geq r_i - \miss_i$.
\end{enumerate}
One can verify that solutions to this program directly
correspond to $\epsilon$-almost-covers for instance~$I$. This
completes the proof.
\end{proof}

\subsection{From Covering Problems to Approval Voting}
In this section we show a relation between several problems regarding
Approval voting and the covering problems studied above. In
consequence, we will explain how our technique can be used for
obtaining fixed-parameter tractability results for these voting
problems.

We model an election as a pair $E = (C,V)$, where $C = \{c_1, \ldots,
c_m\}$ is a set of candidates and $V = (v_1, \ldots, v_n)$ is a
collection of voters. Each voter is represented through his or her
preferences. For the case of Approval voting, each voter's preferences
take the form of a set of candidates approved by this voter.  The
candidate(s) receiving the most approvals are the winner(s). In other
words, we assume the nonunique-winner model (if several candidates
have the same number of approvals, then we view each of them as
winning).  We write $\score_E(c_i)$ to denote the number of voters
approving~$c_i$ in election~$E$.  We refer to elections that use
Approval voting and represent voter preferences in this way as
approval elections.
%
%
In a weighted election, voters also have integer weights in addition
to their preferences. A voter $v$ with weight $\omega(v)$ counts as
$\omega(v)$ copies of an unweighted voter.\footnote{There is a name
  clash between the literature on covering problems and that on
  elections. In the former, ``weights'' refer to what the voting
  literature would call ``prices.''  Weights of the voters are modeled
  as multiplicities of the elements in the multisets. We kept the
  naming conventions from the respective parts of the literature to
  make our results more accessible to researchers from both
  communities.}

We are interested in the following three problems.

\begin{definition}[Bartholdi et al.~\cite{BTT92}, Faliszewski et al.~\cite{FHH09,MF16}]
  In each of the problems Approval-\textsc{\$Bribery} (priced bribery), Approval-\textsc{\$CCAV}
  (priced control by adding voters), and Approval-\textsc{\$CCDV}
  (priced control by deleting voters), we are given an
  approval election $E = (C,V)$ with $C = \{p, c_1, \ldots, c_m\}$
  and $V = (v_1, \ldots, v_n)$, and an integer budget $B$. In
  each of the problems the goal is to decide whether it is possible to
  ensure that $p$ is a winner, at a cost of at most~$B$. The problems
  differ in the allowed actions and possibly in some
  additional parts of the input:
  \begin{enumerate}[leftmargin=1.2em,labelsep=0.5em]\fixlist

  \item In Approval-\textsc{\$Bribery}, for each voter $v_i$, $1 \leq
    i \leq n$, we are given a nonnegative integer price~$\pi_i$; for
    this price we can change $v_i$'s approval set in any way we
    choose.
  \item In Approval-\textsc{\$CCAV} (CCAV stands for ``Constructive 
    Control by Adding Voters'') we are given a collection $Q =
    (q_1, \ldots, q_{n'})$ of additional voters.  For each additional
    voter $q_i$, $1 \leq i \leq n'$, we also have a nonnegative integer
    price $\pi_i$ for adding~$q_i$ to the original election.
  \item In Approval-\textsc{\$CCDV} (CCDV stands for ``Constructive 
    Control by Deleting Voters''), we have a
    nonnegative integer price $\pi_i$ for removing each voter $v_i$ from the
    election.
  \end{enumerate}

  \noindent In the weighted variants of these problems (which we denote by
  putting ``\textsc{Weighted}'' after ``Approval''), the input
  elections (and all the voters) are weighted; in particular, each voter $v$ has an
  integer weight $\omega(v)$.
  The unpriced variants of these problems (denoted by omitting the
  dollar sign from their names) are defined identically, except that
  all prices have the same unit value.
\end{definition}


The above problems are, in essence, equivalent to certain covering
problems. Briefly put, the relation between \textsc{WMM} and various
election problems (as those defined above) is that the universe
corresponds to the candidates in the election, the multisets
correspond to the voters, and the covering requirements depend on
particular actions that we are allowed to perform.

\begin{construction}\label{const:1}
  Consider an instance of Approval-\textsc{\$CCDV} with election $E =
  (C,V)$, where $C = \{p, c_1, \ldots, c_m\}$ and $V = (v_1,
  \ldots, v_n)$, with prices $\pi_1, \ldots, \pi_n$ that one needs to pay to
  the respective voters in order to convince them not to
  participate in the election, and with budget $B$. We can express
  this instance as an instance of \textsc{Weighted Multiset
    Multicover} as follows. For each voter $v_i$ not approving $p$,
    we form a multiset $S_i$ with weight $\pi_i$ that includes
  exactly the candidates approved by $v_i$, each with multiplicity
  exactly one. For each candidate $c_i$, $1 \leq i \leq m$, we set its
  covering requirement to be $\max(\score_E(c_i)-\score_E(p),0)$. 
  There is a way to ensure $p$'s victory by
  deleting voters of total cost at most $B$ if and only if it is
  possible to solve the presented instance of \textsc{Weighted
    Multiset Multicover} with budget $B$.
\end{construction}

Naturally, we do not use the full generality of \textsc{WMM} in
\autoref{const:1}; in fact, we provide a reduction to \textsc{Weighted
  Set Multicover}, where the multiplicities of input multisets are
either 0 or~1.  This is important since
\autoref{prop:wmsmcpara} says that \textsc{WMM} is $\np$-hard
even for a single element in the universe.  From the viewpoint of
voting theory, it is also interesting to consider 
\textsc{Uniform Multiset Multicover}, where for each multiset $S_i$ in the
input instance there is a number $t_i$ such all elements belonging to
$S_i$ have multiplicity either equal to zero or to $t_i$.  Using an
argument similar to that used in \autoref{const:1}, it is easy to show
that \textsc{Uniform Multiset Multicover} is, in essence, equivalent
to Approval-\textsc{Weighted-CCDV}.

In \autoref{const:1} we have considered Approval-\textsc{\$CCDV}
because, among our problems, it is the most straightforward one to
model via a covering problem. Nonetheless, constructions with similar
flavor are possible both for Approval-\textsc{\$CCAC}
and for Approval-\textsc{\$Bribery}.
Formally, we have the following result.
\begin{proposition}\label{prop:vot2cover}
  \textsc{Approval-\$CCAV}, \textsc{Approval-\$CCDV}, \textsc{Approval-\$Bribery},
  \textsc{Approval-Weighted-CCAV}, and \textsc{Approval-Weighted-CCDV} are 
  fixed-parameter tractable when parameterized by the number of candidates.
%
%
\end{proposition}

\begin{proof}
  We describe for each voting problem either a reduction to
  \textsc{Weighted Set Multicover} or to \textsc{Uniform Multiset
    Multicover}.  Formally, we either use standard many-one reductions
  or very simple special cases of Turing-reductions: In an outer loop,
  we iterate through certain values, which give an additional hint on
  how the solution looks like (we refer to it as ``guessing''), and
  then resolve the remaining problem by a transformation to one of the
  two covering problems.  We finally answer yes if one of the covering
  instances was a yes-instance.  In our reductions, the universe
  set~$U$ is always identical to the candidate set~$C$, but the
  covering requirements, the family~$\calS$ of the (multi)sets, the
  weights, and the prices differ.

 \paragraph{Approval-\textsc{\$CCDV}.} See \autoref{const:1}.
 
 \paragraph{Approval-\textsc{\$Bribery}.}
 Consider an instance of Approval-\textsc{\$Bribery} with election $E
 = (C,V)$, where $C = \{p, c_1, \ldots, c_m\}$ and $V = (v_1, \ldots,
 v_n)$, with prices $\pi_1, \ldots, \pi_n$ for changing the voter's
 approval set, and with budget $B$.  Observe that
 Approval-\textsc{\$Bribery} is very similar to
 Approval-\textsc{\$CCDV}, because we can assume without loss of
 generality that each bribed voter finally approves only
 candidate~$p$.  However, the decisive difference is that we do not
 know the final number of approvals that~$p$ will get because this
 depends on the given budget~$B$, on the prices of the voters, and on
 how many bribed voters already approved~$p$ (but together with some
 other candidates).  We circumvent this lack of knowledge by guessing
 the number~$\ell$ of additional approvals $p$~obtains through the
 bribery process.  This also gives us the
 score~$s^*:=\score_{E}(p)+\ell$ of~$p$ in the final election
 (containing the $\ell$~bribed voters).  Now, we have to ensure (i)
 that $p$~really obtains the guessed score and (ii) that all other
 candidates which originally have a higher score lose enough approvals
 through the bribery process.  We can express this as an instance of
 \textsc{Weighted Multiset Multicover} as follows.  For each voter
 $v_i \in V$, we form a multiset $S_i$ with weight $\pi_i$ that
 includes all the candidates approved by $v_i$, each with multiplicity
 exactly one, as well as candidate~$p$ also with multiplicity one if
 and only if $v_i$ \emph{does not} approve~$p$.  For each candidate $c
 \in C \setminus \{p\}$, we set its covering requirement to be
 $\max(\score_E(c)-s^*,0)$.  For~$p$ we set the covering requirement
 to~$\ell$.  It is easy to see that there is a way to ensure $p$'s
 victory by adding voters of total cost at most $B$ if and only if it
 is possible to solve the presented instance of \textsc{Weighted
   Multiset Multicover} with budget~$B$.  Since the constructed
 instance is, in fact, an instance of \textsc{Weighted Set
   Multicover}, we obtain an $\fpt$ algorithm.

 \paragraph{Approval-\textsc{\$CCAV}.}
 Consider an instance of Approval-\textsc{\$CCAV} with election $E =
 (C,V)$, where $C = \{p, c_1, \ldots, c_m\}$, $V = (v_1, \ldots,
 v_n)$, and $Q = (q_1, \ldots, q_{n'})$, with prices $\pi_1, \ldots,
 \pi_{n'}$ that one needs to pay to the respective voters from~$Q$ in
 order to convince them to participate in the election, and with
 budget $B$.  It is never useful to add a voter that does not approve
 candidate~$p$.  Adding a voter~$w$ (who approves~$p$) to the election
 has one decisive effect: it decreases the score difference between
 candidate~$p$ and each candidate that is not approved by~$w$.  Hence,
 we can express this instance as an instance of \textsc{Weighted
   Multiset Multicover} as follows.  For each voter $q_i \in Q$
 approving $p$, we form a multiset $S_i$ with weight $\pi_i$ that
 includes exactly the candidates not approved by $q_i$, each with
 multiplicity exactly one.  For each candidate $c \in C$, we set its
 covering requirement to be $\max(\score_E(c)-\score_E(p),0)$. It is
 easy to see that there is a way to ensure $p$'s victory by adding
 voters of total cost at most $B$ if and only if it is possible to
 solve the presented instance of \textsc{Weighted Multiset Multicover}
 with budget $B$. Since the constructed instance is, in fact, an
 instance of \textsc{Weighted Set Multicover}, we obtain an $\fpt$
 algorithm.
 
 \paragraph{\textsc{Approval-Weighted-CCDV} and \textsc{Approval-Weighted-CCAV}.}
 By analogous arguments as above we do the same construction as for
 Approval-\textsc{\$CCDV} (resp.\ Approval-\textsc{\$CCAV}) except
 that
 \begin{inparaenum}[(i)]
 \item we omit the weights of the multisets, and
 \item we set the multiplicity for each element in the multiset to the
   weight of the corresponding voter.
 \end{inparaenum}
 In consequence, we obtain instances of \textsc{Uniform Multiset
   Multicover}, which can be solved in $\fpt$ time.
\end{proof}




On the other hand, it is either shown explicitly by \citet{FHH09} or
follows trivially that when the problems from the above proposition
have both prices and weights, then they are $\np$-hard already for two
candidates (that is, they are $\paranp$-hard with respect to the
number of candidates).



\section{Further Generalizations of the Results Related to Voting}\label{sec:generalizations}

We now consider the ordinal model of elections, where each
voter's preferences are represented as an order, ranking the
candidates from the most preferred one to the least preferred one. For
example, for $C = \{c_1,c_2,c_3\}$, vote $c_1 \pref c_3 \pref c_2$ means that
the voter likes $c_1$ best, then $c_3$, and then
$c_2$.

There are many different voting rules for the ordinal election model.
Here we concentrate only on scoring rules.
A scoring rule for $m$ candidates is a nondecreasing vector $\alpha = (\alpha_1,
\ldots, \alpha_m)$ of integers. Each voter gives $\alpha_1$ points to
his or her most preferred candidate, $\alpha_2$ points to the second
most preferred candidate, and so on. Examples of scoring rules include
the Plurality rule, defined through vectors of the form $(1,0, \ldots,
0)$, $k$-Approval, defined through vectors with $k$~ones followed by
$m-k$~zeroes, and Borda count, defined through vectors of the form $(m-1,
m-2, \ldots, 0)$.

For each voting rule $\calR$ in the ordinal model, it is
straightforward to define $\calR$-\textsc{\$CCAV},
$\calR$-\textsc{\$CCDV}, and $\calR$-\textsc{\$Bribery}.
Using our new framework, we obtain the following result.

\begin{theorem}\label{thm:ordinal}
  For every voting rule $\calR$ for which winner determination can be
  expressed through a set of integer linear inequalities over
  variables that indicate how many voters with each given preference
  order are in the election, $\calR$-\textsc{\$CCAV},
  $\calR$-\textsc{\$CCDV}, and $\calR$-\textsc{\$Bribery} are 
  fixed-parameter tractable when parameterized by the number of candidates.
\end{theorem}

\begin{proof}
The proof follows the same structure as that of \autoref{thm:wsmc-fpt}.
We will present the proof only for $\calR$-\textsc{\$CCDV};
the other cases follow by applying the same approach.
Let us consider an instance of $\calR$-\textsc{\$CCDV} with the set of candidates $C = \{p, c_1, \ldots, c_m\}$,
the collection of voters $V = (v_1, \ldots, v_n)$, prices $\pi_1, \ldots, \pi_n$, and with budget~$B$.

Let $X$ be the set of integer variables which indicate how many voters with each given preference
order are in the election. By $x_{\sigma}$ we denote the variable from~$X$ which corresponds to
the preference ranking~$\sigma$.
Clearly, the size of $X$ is upper-bounded by~$m!$, i.e., by a function of the number of the candidates. 
Let $S$ be the set of inequalities over variables from $X$ that encode that $p$ is a winner in the election.

We construct an integer program with convex transformations as
follows. For each preference order $\sigma$ we introduce one integer
variable $c_{\sigma}$.  Intuitively, this variable indicates how many
voters with the preference order $\sigma$ we need to remove from the
election.  Additionally, we introduce a function $f_\sigma$ such that
$f_\sigma(c_{\sigma})$ is the total price of $c_{\sigma}$ least
expensive voters whose preference order is~$\sigma$. For each
$x_{\sigma} \in X$ we replace $x_{\sigma}$ in $S$ with the number of voters
from $V$ whose preference ranking is $\sigma$ minus $c_{\sigma}$, and
we add a constraint enforcing that this difference is greater or equal
to zero.  Finally, we add the budget constraint $\sum_{\sigma}
f_\sigma(c_{\sigma}) \leq B$. It is apparent that our ILP is feasible
if and only if the answer to the original instance is ``yes''.  We
solve it in $\fpt$ time via \autoref{thm:general_theorem}.
\end{proof}

For a more detailed description of the class of voting rules where
``winner determination can be expressed through integer linear
inequalities,'' we point the reader to the works of Dorn and
Schlotter~\cite{DS12} or of Faliszewski et al.~\cite{FHH11}.  In
particular, \autoref{thm:ordinal} applies to all scoring rules.
This, and the results from the previous section, resolves an issue
dating back to the work of Faliszewski et al.~\cite[Theorem 4.4,
Theorem 4.13; the conference version of their work was published in
2006]{FHH09}, who have shown that \textsc{\$Bribery} is in $\xp$ for
approval voting and for scoring protocols (for the parameterization by
the number of candidates).\footnote{They did not speak of
  $\xp$-membership explicitly, but this is exactly what they have
  shown. Bredereck et al.~\cite{BCFetal14} popularized the issue of
  resolving if \textsc{\$Bribery} problems parameterized by the number
  of candidates are in $\fpt$ or are $\mathrm{W}[\cdot]$ hard, leading
  in particular to our solution and to the slightly later work of
  Kouteck\'y et al.~\cite{KKM17}.}  Until our work, the exact
parameterized complexity of these problems was unknown.


Our framework also allows to partially resolve an open problem posed by Bredereck et
al.~\cite{BCFNN16} regarding \textsc{Shift Bribery}.
In this problem we are given an election and a preferred candidate~$p$, and the goal
is to ensure $p$'s victory by shifting $p$ forward in some of
the votes (the cost of each shift depends on the number of positions
by which we shift $p$). Under the ``sortable prices assumption'', 
voters with the same preference orders can be sorted so that if
voter~$v'$ precedes voter~$v''$, then we know that shifting $p$ by
each given number of positions~$i$ in the vote of~$v'$ costs at most
as much as doing the same in the vote of~$v''$.  Using this
assumption, we obtain the following result (all-or-nothing prices are
a special case of sortable prices where we always shift $p$ to the top
of a given vote or we leave the vote unchanged).

\begin{theorem}\label{thm:sb}
  For Borda (and for Maximin and Copeland voting rules),
  \textsc{Shift Bribery} for sortable price functions and for
  all-or-nothing price functions is fixed-parameter tractable when parameterized by the
  number of candidates.
\end{theorem}

\citet{BCFNN16} gave an $\fpt$ approximation scheme for the problems
from \autoref{thm:sb}; we can use part of their algorithm and
apply our new framework in order to derive an exact and not only 
approximate solution.  Their algorithm rephrases the problem and
then applies a bounded search through the solution space.  We can use
their rephrasing but replace the search by solving a \eMIP instance.
We omit technical details since recently \citet{KKM17} showed
fixed-parameter tractability of the \textsc{Swap Bribery} problem
parameterized by the number of candidates, as part of a very general
result using the $n$-fold IP technique. \textsc{Swap Bribery} is a
generalization of \textsc{Shift Bribery}, so the result of Kouteck\'y
et al.~\cite{KKM17} is stronger than the one given above.

\section{Discussion \& Outlook}\label{sec:conclusions}

We have proposed an extension of Lenstra's famous result for solving
ILPs.  In our extended formulation, one can replace any integer
variable with its simple piecewise linear transformation---this
transformation needs to be either convex or concave, depending on the
position of the variable in the ILP.  We have shown that such extended
ILPs can still be solved in $\fpt$ time with respect to the number of
integer variables, as long as there are at most polynomially many
pieces.

We have demonstrated several applications of our general result which
relate to classic covering problems and to selected voting problems.
Most notably, we have proven that \textsc{Weighted Set Multicover} is
fixed-parameter tractable when parameterized by the number of elements
to cover. Further, building upon our general result, but using a more
technically involved argument, we have proved the existence of an
$\fpt$ approximation scheme for \textsc{Multiset Multicover}, also for
the parameterization by the number of elements.  We have also explained
how our general results can be used in studies on control and bribery
in elections---we have shown that certain variants of these problems
are in $\fpt$ when parameterized by the number of candidates. In
particular, we have resolved the parameterized complexity of some
problems that were open for the last ten years or so.


Our paper leads to several possible directions for future work.
First, unfortunately, while Lenstra's algorithm is a very powerful
tool for proving FPT membership, it might be too slow in practice.
Thus, as pointed out by Bredereck et al.~\cite{BCFetal14}, each time
an FPT result is achieved through an application of Lenstra's result,
it is natural to ask whether one can derive the same result through a
direct, combinatorial algorithm.  Coming up with such a direct
algorithm usually seems very difficult.  In practice, one would
probably not use Lenstra's algorithm for solving MIPs, but, instead,
one of the off-the-shelf optimized heuristics. In the conference
version of this paper~\cite{BFNST15} we provided a preliminary
empirical comparison of the running times of the MIP-based algorithm
(using an off-the-shelf MIP solver instead of Lenstra's algorithm) and
an ILP-based algorithm that reduces our problems directly to integer
linear programming (basically without ``exploiting'' the
parameter). Our results suggest that $\fpt$ algorithms based on
solving MIPs can be very efficient in practice.  A more thorough
experimental analysis of these and similar questions would help in understanding
the real power and limitations of the techniques based on
MIPs, thus we believe it is an important research direction.  Second,
our work advances a fairly modest literature on FPT approximation schemes.
It would be very interesting to further explore the (practical) relevance of
such algorithms. 

\subsection*{Acknowledgments}
We thank Martin Kouteck\'{y} for his extremely helpful input and feedback 
concerning related work.

Robert Bredereck was from September 2016 to September 2017 on
postdoctoral leave at the University of Oxford, supported by the DFG
fellowship BR 5207/2 and at project start in early 2015 supported by DFG project PAWS (NI 369/10).
Piotr Skowron was supported by a Humboldt Research Fellowship for Postdoctoral Researchers.
Piotr Faliszewski was supported by DFG project PAWS (NI 369/10) during his stay at TU
Berlin and by AGH University grant 11.11.230.337 (statutory research)
afterward.  Nimrod Talmon was supported by DFG, Research Training
Group ``Methods for Discrete Structures'' (GRK 1408), while the author
was affiliated with TU Berlin, where most of the work was done. This
work was also partly supported by COST Action IC1205 on Computational
Social Choice.



\bibliographystyle{abbrvnat}
\bibliography{bibliography}

\end{document}